\documentclass[12pt]{article}
\usepackage[utf8]{inputenc}
\usepackage[top=1in, bottom=1in, left=0.75in, right=0.75in]{geometry}
\usepackage{setspace}
\usepackage{natbib}
\usepackage{amsmath}
\usepackage{amssymb}
\usepackage[dvipsnames]{xcolor}
\usepackage{graphicx}
\usepackage{amsthm}
\usepackage{mathtools}
\usepackage{algpseudocode}
\usepackage{algorithm}
\usepackage{appendix}
\usepackage{authblk}
\usepackage{subfigure}
\usepackage{bbm}
\usepackage{adjustbox}
\usepackage{booktabs}
\newtheorem{theorem}{Theorem}[section]

\newtheorem{lemma}[theorem]{Lemma}

\newtheorem*{remark}{Remark}
\newtheorem{assumption}{Assumption}

\newcommand{\Rem}{\textnormal{Rem}}
\newcommand{\calM}{\mathcal{M}}

\newcommand{\calT}{\mathcal{T}}
\newcommand{\calC}{\mathcal{C}}
\newcommand{\calW}{\mathcal{W}}
\newcommand{\EE}{\mathbb{E}}

\newcommand{\rconly}{^{\text{rc only}}}
\newcommand{\eff}{^{\textnormal{eff}}}
\newcommand{\I}[1]{\mathbbm{1}(#1)}

\title{Efficient Inference for Time‐to‐Event Outcomes by Integrating Right‐Censored and Current Status Data}
\author{Xiudi Li$^{1}$\footnote{These authors contributed to this work equally.} \ and Sijia Li$^{2*}$}
\date{$^1$ Division of Biostatistics, University of California, Berkeley \\
$^2$ Department of Biostatistics, University of California, Los Angeles}

\begin{document}

\maketitle
\begin{abstract}
We propose a semiparametric data fusion framework for efficient inference on survival probabilities by integrating right-censored and current status data. Existing data fusion methods focus largely on fusing right-censored data only,  while standard meta-analysis approaches are inadequate for combining right-censored and current status data, as estimators based on current status data alone typically converge at slower rates and have non-normal limiting distributions. In this work, we consider a semiparametric model under exchangeable event time distribution across data sources. We derive the canonical gradient of the survival probability at a given time, and develop one-step estimators along with the corresponding inference procedure. Specifically, we propose a doubly robust estimator and an efficient estimator that attains the semiparametric efficiency bound under mild conditions. Importantly, we show that incorporating current status data can lead to meaningful efficiency gains despite the slower convergence rate of current status–only estimators. We demonstrate the performance of our proposed method in simulations and discuss extensions to settings with covariate shift. We believe that this work has the potential to open new directions in data fusion methodology, particularly for settings involving mixed censoring types.
\end{abstract}

\section{Introduction}
\label{s:introduction}
A major challenge in survival analysis is that event times are not always fully observed. Censoring is especially common and gives rise to two frequently encountered data types: right-censored and current status data. Right-censored data arise when for some individuals, the exact event time is unknown but it is known to occur after some censoring time. Right censoring is common in longitudinal cohort studies and clinical trials where subjects are followed over time, and occurs when a subject is lost to follow-up or the study ends prior to the event of interest. Current status data arise when each subject is assessed only once at some inspection time to determine whether the event has occurred, without observing the exact timing. As an extreme form of interval censoring, current status data provide a snapshot in time instead of a longitudinal record, and are often encountered in cross-sectional surveys \citep{diamond1986proportional,keiding1991age}, screening programs \citep{jewell1990statistical}, and epidemiological studies \citep{becker2017analysis}. 

Importantly, both right-censored data and current status data may arise side-by-side within the same research context. For example, in cancer epidemiology, patients enrolled in prospective cohorts are followed over time until a diagnosis or censoring, generating right-censored data. Meanwhile, population-based cancer registries such as the Surveillance, Epidemiology, and End Results (SEER) Program, or cross-sectional surveys \citep{caplan1995use,iezzoni2020cross}, collect cancer prevalence data by recording whether individuals have been diagnosed at a specific point in time, resulting in current status data. In infectious diseases research, prospective surveillance of at-risk individuals in clinical trials captures infection times subject to right censoring \citep{polack2020safety,cohen2011prevention}, whereas periodic seroprevalence surveys \citep{havers2020seroprevalence,wolock2025investigating} provide current status data on infection in the community. Notably, both sources of data contain valuable information on the time to event outcome of interest. 

Under minimal model assumptions, the statistical properties of nonparametric estimators for the survival function constructed from right-censored and current status data differ substantially. For right-censored data under independent or conditional independent censoring, widely used methods such as the Kaplan-Meier estimator \citep{kaplan1958nonparametric}, inverse probability of censoring weighted estimators \citep{robins1992recovery}, and augmented inverse probability of censoring weighted estimators \citep{van2003unified,tsiatis2006semiparametric}, are root-n consistent and asymptotically normal. In contrast, when using only current status data without strong parametric assumptions, the survival function is not pathwise differentiable in general. As a result, estimators typically converge at a slower rate with a non-Gaussian limiting distribution. \cite{ayer1955empirical} and \cite{groeneboom2001computing} proposed cubic root-n rate nonparametric maximum likelihood estimators when the inspection time is marginally independent of the event time. However, such condition may fail in practice, and \cite{lagakos1988use} discussed several examples in which these two times are correlated. To handle such dependence, many existing methods impose parametric assumptions on the underlying data generating distributions \citep{zhang2005statistical} or model the dependence using a pre-specified copula function \citep{zheng1995estimates,wang2012nonparametric,ma2015sieve}. More recently, \cite{wolock2025investigating} studied the setting where the inspection time and event time are conditionally independent given baseline covariates, an assumption that is often more plausible in practice. Their proposed causal isotonic regression estimator is cubic-root consistent and tends to a scaled Chernoff distribution asymptotically \citep{groeneboom2001computing}.

Despite the differences in estimation efficiency, both data types are often available for estimating the survival function. This motivates a key question: how can we efficiently combine the information from both right-censored and current status data to achieve better inference than using either alone? To the best of our knowledge, there is no established nonparametric or semiparametric framework for integrating these two types of data when the corresponding estimators differ fundamentally in convergence rates and asymptotic behaviors. Standard meta-analysis methods, such as inverse-variance weighting, are inadequate in this context due to the non-Gaussian limiting distribution. While alternatives based on the confidence distribution \citep{liu2022nonparametric}, repro-sampling methodologies \citep{xie2022repro}, and more general Bayesian and Fiducial approaches may be applicable, these methods have not been formally studied in our current context.  Moreover, as sample sizes of the two data types grow proportionally, standard combination methods tend to place all weights on the estimator converging at root-n rate, rendering the contribution of current status data asymptotically negligible. Another line of recent works on fusing survival data has been limited to integrating multiple right-censored datasets \citep{gao2024doubly,wen2025multi,liu2025targeted}. Therefore, efficiently fusing right-censored and current status data in a way that fully leverages both sources remains an important open problem.

In this work, we develop a principled approach to integrating right-censored and current status data for estimating marginal survival probability at a given time point of interest. We consider the practical setting under conditionally uninformative censoring and inspection times. When the joint distribution of event times and covariates is exchangeable between the two types of data, we establish the semiparametric efficiency bound for estimating survival probabilities and propose a one-step estimator that achieves this efficiency bound. Two key implications of our results are as follows. First, enriching a right-censored sample with current status data preserves the desirable root-n convergence rate while yielding additional efficiency gains over using right-censored data alone. Second, augmenting current status data with right-censored observations drastically improves efficiency and raises the convergence rate from $n^{-1/3}$ to $n^{-1/2}$. Moreover, we also propose another one-step estimator that enjoys the double robustness property --- one can employ misspecified working models for either the conditional distribution of the event time or the conditional distributions of the censoring time and inspection time, and the resulting estimator remains consistent. We further extend our method to settings with covariate shift. While our main results focus on marginal survival probabilities and two data sources, the methodology naturally generalizes to multiple data sources or to causal parameters such as average treatment effects measured as differences in survival probabilities. We believe this work has the potential to open a new line of exciting research around data fusion with mixed censoring types. 

The rest of this paper is organized as follows. In the remainder of this section, we review relevant literature and introduce key notations. In Section~\ref{sec: fully observed and current status}, we consider the fusion of fully observed event-time data with current status data, and derive the canonical gradient for the survival probability. In doing so, we introduce an integral equation central to our proposed approach whose solution is a key component of the gradients used to construct our estimators for both the fully observed and right-censored settings. In Section~\ref{sec: right censored and current status}, we study the fusion of right-censored and current status data and derive the corresponding canonical gradient. In Section~\ref{sec: one-step estimation}, we introduce the one-step estimators and establish their statistical properties. In Section~\ref{sec: experiments}, we demonstrate the proposed method through simulation studies. In Section~\ref{sec: discussion}, we conclude with a discussion.

\subsection{Related literature}
Our work builds on the growing literature on data fusion. \cite{li2023efficient} developed a general semiparametric framework for estimating pathwise-differentiable target estimands from multiple data sources. \cite{graham2024towards} extended this approach to accommodate situations where the conditional distributions of data sources are aligning under different factorizations of the target distribution. Other works have tailored these approaches specifically to right-censored outcomes: for example,  \cite{liu2025targeted} studied the estimation for causal survival differences, and  \cite{gao2024doubly}  introduced a selective borrowing approach for estimating the restricted mean survival time. However, our setting falls outside these prior works due to two unique features: (1) the event time of interest is partially observed in all data sources, and (2) the available data sources exhibit different censoring types. 

To address these gaps, our approach leverages tools from semiparametric efficiency theory. Central to this framework is the gradient (or influence function), which characterizes the sensitivity of the estimand to local perturbations in the data distribution and enables the construction of regular and asymptotically linear estimators. By identifying the canonical gradient, that is, the gradient with the smallest variance, we can construct estimators that are both consistent and efficient, achieving the semiparametric efficiency bound. In this paper, we employ the one-step estimation strategy, although alternative approaches such as targeted minimum loss based estimation \citep{van2011targeted} could also be used. We refer the readers to \cite{bickel1993efficient,ibragimov2013statistical,bickel1982adaptive} for overviews of semiparametric efficiency theory and one-step estimation.

\subsection{Notations}
Throughout, we use capital letters to denote random variables and the corresponding lowercase letters for their realizations. We condition on lowercase letters in expectations to indicate conditioning on a random variable taking a specific value. For a generic distribution $\nu$, we let $\EE_\nu$ denote the expectation operator under $\nu$.

We use $W \in \mathcal{W} \subset \mathbb{R}^d$ to denote the covariate vector and $T \in \mathbb{R}$ to denote the time to the event of interest. In the right-censored data, we use $R \in \mathbb{R}$ to denote the censoring time, and we observe $(W, \Delta_R, Y)$ where $\Delta_R=\mathbbm{1}(T\leq R)$ is the event indicator and $Y= \min(T,R)$ is the observed time. In the current status data, we let  $C$ denote the inspection time, and observe $(W, \Delta_C,C)$ where $\Delta_C= \mathbbm{1}(T\le C)$ is a binary indicator of whether the event has occurred by $C$. We let $\calC$ denote the support of $C$. In addition, let $S$ denote the type of data source, with $S=0$ for current status data and $S=1$  for right-censored data. The observed data unit can then be compactly written as $X=(S, W,(1-S)\Delta_C, (1-S)C,S\Delta_R,SY)$, whose distribution is denoted by $P$.  

\section{Integrating fully observed event time and current status data}
\label{sec: fully observed and current status}
We first consider the case where the event time $T$ is fully observed in data $S=1$. This scenario arises, for instance, in infectious disease studies where outbreak investigations provide exact infection times for a subset of cases with no right-censoring, while a parallel community survey collects only whether individuals have been infected by a certain survey date. Under such settings, the event time $T$ is fully observed for $S=1$ and $(C,\Delta_C)$ is observed for $S=0$. We refer to this as the ideal fused-data model, as it is absent of right-censoring. In what follows, we first introduce relevant assumptions, and then derive the corresponding tangent space and the canonical gradient of survival probabilities with respect to this model. This result could be of independent interest, but more importantly we introduce an integral equation whose solution $h^*$ is central to our later developments. For this purpose, assuming in this section that the observed data unit is $X^I = (S, W,ST,(1-S)\Delta_C,(1-S)C)$, whose distribution is denoted as $P^I$.

We impose the following assumptions for our main results.
\begin{assumption}[Exchangeability] $(T,W) \perp S$. \label{cond:exchangeability}
\end{assumption}
\noindent Assumption~\ref{cond:exchangeability} requires the joint distribution of time to event $T$ and covariates $W$ to be the same for $S=0$ and $S=1$. In Appendix~\ref{app: extensions}, we relax this assumption to conditional exchangeability that only requires $T \perp S \mid W$, and generalize our results to settings with covariate shifts. However, for presentation clarity, we decided to introduce our main results under the more restrictive assumption. In practice, such exchangeability may be reasonable if the two sets of data are collected from the same population within a similar time frame. 

\begin{assumption}[Conditional independent inspection times] $T \perp C \mid W, S=0$. \label{cond:uninformative inspection}
\end{assumption}

\noindent Assumption~\ref{cond:uninformative inspection} relaxes the marginal independence between inspection times $C$ and event times $T$, and only requires independence within strata defined by covariates $W$. Such independence condition is likely to hold if $W$ contains all covariates that inform the relationship between the inspection times and event times. 

Under Assumptions~\ref{cond:exchangeability} and \ref{cond:uninformative inspection}, the distribution $P^I$ of $X^I$ is fully determined by (1) the marginal distribution of $S$, denoted as $\Pi$, (2) the marginal distribution of $W$, denoted as $P_W$ with denisty $p_W$, (3) the conditional distribution of the event time given covariates, $P_{T|W}$ with conditional CDF $F_{T|W}$ and density function $f_{T|W}$, (4) the conditional distribution of inspection time given covariates $P_{C|W}$ with CDF $G_{C|W}$ and density $g$. That is, $P^I = \mathbb{P}(\Pi,P_W,P_{T|W},P_{C|W})$ for some mapping $\mathbb{P}$. Let $\calM^I$ denote a semiparametric model containing all induced distributions of $X^I$ satisfying Assumptions~\ref{cond:exchangeability} and \ref{cond:uninformative inspection}, that is, $\calM^I = \{\mathbb{P}(\Pi,\tilde P_W,\tilde P_{T|W},\tilde P_{C|W})\}$, with $\Pi$ known by design and no additional assumptions imposed on the remaining three component distributions. Then, $P^I \in \calM^I$. 

Finally, we impose the following assumption on the support of inspection time.
\begin{assumption}[Inspection window positivity]\label{cond:support of C}
The support of $C$ under $P$, denoted by $\calC$, is compact. Moreover, there exists some constant $\zeta >0$ such that $F_{T|W}(c|w) \in [\zeta, 1-\zeta]$ for all $c \in \mathcal{C}$ and $w \in \mathcal{W}$.
\end{assumption}
\noindent Assumption~\ref{cond:support of C} is a technical assumption requiring that the support of the inspection time is bounded and lies within the support of the event time. Such a condition is often reasonable in practice, as inspection is typically scheduled at a time when the event could plausibly occur. Furthermore, if needed, one can restrict to the subset of current status observations whose inspection times fall within a chosen compact set contained in the support of $T$, thereby ensuring this assumption holds. While such pre-processing may alter the covariate distribution $P_{W|S=0}$, we show in Appendix~\ref{app: extensions} that our framework accommodates covariate shifts naturally.

Under Assumptions~\ref{cond:exchangeability} to \ref{cond:support of C}, the likelihood of one observation is given by:
\begin{align*}
    p_W(w)f_{T\mid W}(t\mid w)^s \left\{g(c\mid w) \left(\int_0^c f_{T\mid W}(u\mid w)du\right)^{\delta_c}\left(1-\int_0^c f_{T\mid W}(u\mid w)du\right)^{1-\delta_c}\right\}^{1-s}. \label{eq:lik_full}
\end{align*}
By considering local perturbations of the distribution $P^I$, we can derive the tangent space at $P^I$ relative to the model $\calM^I$ as given in the following lemma. 
\begin{lemma}\label{lemma: tangent space FO and CS}
The tangent space at $P^I$ with respect to the model $\calM^I$ takes the following form
\begin{align}
    \calT^I &= \overline{\textnormal{span}}\Bigg\{h_W(w) + sh(t,w) + (1-s)h_C(c,w) + \frac{(1-s)(\delta_C - F_{T|W}(c|w))}{F_{T|W}(c|w)(1-F_{T|W}(c|w))}\int_0^c f_{T|W}(u|w)h(u,w)du, \nonumber \\
    &\quad \quad \quad h_W \in L_2^0(P_W), h(t,w) \in L_2^0(P_{T|W}), h_C(c,w) \in L_2^0(P_{C|W})\Bigg\},
\end{align}
where $\overline{\textnormal{span}}$ denotes the closure of the linear span of a set.
\end{lemma}

We aim to estimate the survival function at a given time point of interest $t^*$. Define a functional $\Psi_{t^*}: \calM^I \rightarrow \mathbb{R}$ such that $\Psi_{t^*}(P^I) = \EE_{P_W}[P_{T|W}(T>t^*|W)]$. For simplicity, we write $\psi_{t^*} = \Psi_{t^*}(P^I)$. Note that the parameter $\psi_{t^*}$ is a function of the distribution $P^I$, as it can be identified using the full data $S=1$ alone and estimated by the empirical mean of $\mathbbm{1}(T>t^*)$. Subsequently, a valid corresponding gradient of $\Psi_{t^*}$ is
\begin{align*}
    \tau_{P^I}^{\text{full only}}: x^I \mapsto \frac{s}{\pi} (\mathbbm{1}(t>t^*) - \psi_{t^*}).
\end{align*}

When $t^*$ lies in the support of $C$, under Assumption~\ref{cond:uninformative inspection}, $\psi_{t^*}$ can also be identified using the current status data alone. Specifically, \cite{wolock2025investigating} showed that, by recasting the survival function as a regression problem under a monotonicity constraint, the estimand $\psi_{t^*}$ can be identified by $1- \EE_{P_W}[\EE_{P^I}[\Delta_C \mid W, C=t^*]]$, and they developed an estimator with cubic-root-n convergence accordingly. However, these approaches --using either fully observed event time or current status data alone -- do not make use of all available information when both data types are present. To quantify the potential gains from integrating both sources, we characterize the semiparametric efficiency bound for $\Psi_{t^*}$ relative to the model $\calM^I$ by deriving its canonical gradient. Let $\pi = \Pi(S=1)$ and $\mu(w) = P_{T|W}(T>t^*|W=w)$.

\begin{lemma}\label{lemma: EIF FO and CS}
The parameter $\Psi_{t^*}: \calM^I \rightarrow \mathbb{R}$ is pathwise differentiable and its canonical gradient is
\begin{equation*}\label{eq: EIF FO and CS}
    \tau_{P^I}: x^I \mapsto sh^*(t,w) + \frac{(1-s)(\delta_C - F_{T|W}(c|w))}{F_{T|W}(c|w)(1-F_{T|W}(c|w))}\int_0^c f_{T|W}(u|w)h^*(u,w)du + \mu(w) - \psi_{t^*},
\end{equation*}
where  $h^*$ is the unique solution to
\begin{equation}\label{eq: integral equation}
    \pi h^*(t,w)  + (1-\pi)  \int_t^\infty \frac{H^*(c,w) g(c|w)}{F_{T|W}(c|w)(1-F_{T|W}(c|w))} dc - \gamma(w) - \mathbbm{1}(t>t^*) + \mu(w) = 0,
\end{equation}
for almost every $(t,w)$, with
\begin{align*}
    H^*(c,w) &= \int_0^c f_{T|W}(t|w)h^*(t,w)dt; \\
    \gamma(w) &= (1-\pi)  \int \int_t^\infty \frac{H^*(c,w) g(c|w)}{F_{T|W}(c|w)(1-F_{T|W}(c|w))} f_{T|W}(t|w)dcdt.
\end{align*}
\end{lemma}

Unlike $\tau_{P^I}^{\text{full only}}$, which relies solely on fully observed event data from $S=1$, $\tau_{P^I}$ leverages both data sources for a more efficient estimation for $\psi_{t^*}$. Instead of the simple indicator $\mathbbm{1}(t > t^*) - \psi_{t^*}$, $\tau_{P^I}$ incorporates a refined function $h^*(t, w)$, as well as an augmentation term that draws on the current status observations. The function $h^*$ plays a central role, serving as the common component that links the contributions from both data sources. This structure reflects the fundamental exchangeability assumption of $\mathcal{M}^I$, which posits a shared event time distribution for both data types.  

\begin{remark}
To further illustrate the source of efficiency gain when incorporating current status data, consider the special case where the distribution of $C|W=w$ is degenerate at some fixed time $t^\dagger$ for all $w \in \calW$. In this case, a current status observation reduces to $\mathbbm{1}(T \leq t^\dagger)$ whose mean can be estimated from the current status sample. This is analogous to having an `external' summary statistic available when estimating the survival probability from the `internal' data with fully observed $T$. Since $\mathbbm{1}(T \leq t^\dagger)$ and $\mathbbm{1}(T > t^*)$ are generally correlated for any $t^\dagger$, incorporating the external summary can improve estimation, as suggested by \citet{hu2022semiparametric}. Although in practice $C$ will typically vary, this simplified scenario nonetheless sheds light on the reason behind the expected efficiency gain.
\end{remark}

It is straightforward to show that \eqref{eq: integral equation} can be equivalently written as
\begin{equation*}
    h^*(t,w) = \frac{\mathbbm{1}(t>t^*) - \mu(w)}{\pi} + \int h^*(s,w)K(t,s\mid w)dF_{T\mid W}(s\mid w),
\end{equation*}
with the kernel function 
\begin{equation*}
    K(t,s\mid w) = \frac{\pi - 1}{\pi} \int_s^\infty \frac{g(c\mid w)}{F_{T\mid W}(c\mid w)(1 - F_{T\mid W}(c\mid w))}\left(\mathbbm{1}(c\geq t)  - F_{T\mid W}(c\mid w)\right) dc.
\end{equation*}
Hence, equation \eqref{eq: integral equation} is a Fredholm equation of the second kind. A square-integrability condition, $\int\int K(t,s|w)^2dF_{T|W}(t|w)dF_{T|W}(s|w) < \infty$, is sufficient to guarantee the existence of a solution to equation \eqref{eq: integral equation}. Moreover, because $h^*$ is obtained via orthogonal projection onto the tangent space, its existence implies uniqueness by the pathwise differentiability of $\Psi_{t^*}$ and the uniqueness of the canonical gradient. Although the analytical form of the solution may be intractable, equation~\eqref{eq: integral equation} can be solved numerically. We discuss several strategies in details in Section~\ref{sec: one-step estimation}.

\section{Integrating right-censored and current status data}
\label{sec: right censored and current status}
We now proceed to the setting where the time-to-event data is subject to right censoring for $S=1$. Recall that in this setting, the observation unit is $X = (S, W, SY, S\Delta_R, (1-S)C, (1-S)\Delta_C)$ with distribution $P$. In addition to Assumptions~\ref{cond:exchangeability} to \ref{cond:support of C}, we impose the following assumption on the censoring mechanism: 
\begin{assumption}[Conditional independent censoring] $T \perp R \mid W,S=1$. \label{cond:uninformative censoring}
\end{assumption}

\noindent Assumption~\ref{cond:uninformative censoring} is a common censoring at random assumption for right-censored data and is a special case of coarsening at random \citep{tsiatis2006semiparametric}. Under Assumptions~\ref{cond:exchangeability}-\ref{cond:uninformative censoring}, the distribution $P$ is uniquely determined by $\Pi$, $P_W$, $P_{T|W}$, $P_{C|W}$, and $P_{R|W}$ which is the conditional distribution of the censoring time given covariates. That is, $P = \mathbb{Q}(\Pi,P_W, P_{T|W}, P_{C|W}, P_{R|W})$ for some functional $\mathbb{Q}$. We consider a semiparametric model for $P$, denoted as $\calM$. Specifically, $\calM = \{\mathbb{Q}(\Pi,\tilde P_W, \tilde P_{T|W}, \tilde P_{C|W}, \tilde P_{R|W})\}$ with the distribution $\Pi$ known by design and no restrictions imposed on the other four component distributions. Note that the observed data unit $X$ can be regarded as a further coarsening of the intermediate observation unit $X^I$ introduced in Section~\ref{sec: fully observed and current status}. Consequently, the distribution $P$ can be regarded as an observed data distribution induced by $P^I$ and $P_{R|W}$ that satisfies the coarsening at random assumption. 

We define a functional $\Phi_{t^*}: \calM \rightarrow \mathbb{R}$ such that $\Phi_{t^*}(P) = \EE_P[\EE_P[\Delta_R/\Gamma(Y\mid W)\mathbbm{1}(Y>t^*)\mid W, S=1]]$, with $\Gamma(t|w) = P_{R|W}(R\geq t |W=w)$. For simplicity, we write $\phi_{t^*} = \Phi_{t^*}(P)$. Under Assumptions~\ref{cond:exchangeability} and \ref{cond:uninformative censoring}, we have $\phi_{t^*} = \psi_{t^*} = P(T>t^*)$. Letting $\lambda_{T|W}$ and $\Lambda_{T|W}$ denote the hazard function and cumulative hazard function of $T$ given $W$ respectively, we can obtain a valid gradient of $\Phi_{t^*}$ using only the right-censored data:
\begin{align*}
    & \tau_P^{\text{rc only}}:\\
    & \quad x \mapsto \frac{s}{\pi} \Big\{\int_0^\infty \frac{\mathbbm{1}(u>t^*) - \EE_{T|W}[\mathbbm{1}(T>t^*)|T\geq u,w]}{\Gamma(u|w)}\left\{d\I{y \leq u,\delta_R =1} -\I{y \geq u} d\Lambda_{T|W}(u|w)\right\} \\
    & \hspace{3em}  + \mu(w) - \phi_{t^*} \Big\}.
\end{align*}
The augmented inverse probability of censoring weighted estimator can be regarded as a one-step estimator based on the gradient $\tau_P\rconly$ using only the right-censored data. Alternatively, we can leverage the additional current status data. Specifically, the following lemma gives a gradient of the target parameter $\Phi_{t^*}$ that is based on the previously derived gradient in Lemma~\ref{lemma: EIF FO and CS} in the case without right censoring.

\begin{lemma}\label{lemma: IF RC and CS}
The following function is a gradient of $\Phi_{t^*}: \calM \rightarrow \mathbb{R}$ at $P$
\begin{align}\label{eq: EIF RC and CS}
    \tau_P: x &\mapsto \nonumber \\
    & s\int_0^\infty \frac{h^*(u,w) - \EE_{T|W}[h^*(T,W)|T\geq u,w]}{\Gamma(u|w)}\left\{d\I{y \leq u,\delta_R =1} - \I{y \geq u} d\Lambda_{T|W}(u|w)\right\} \nonumber \\
    &\quad + \frac{(1-s)(\delta_C-F_{T|W}(c|w))}{F_{T|W}(c|w)(1-F_{T|W}(c|w))}\int_0^c f_{T|W}(u|w)h^*(u,w)du + \mu(w) - \phi_{t^*},
\end{align}
with $h^*$ defined as in Lemma~\ref{lemma: EIF FO and CS}.
\end{lemma}

The gradient $\tau_P$ resembles both $\tau_{P^I}$ and $\tau_P^{\text{rc only}}$ but in different ways. $\tau_P$ and $\tau_{P^I}$ share the same component that corresponds to $s=0$, which can be viewed as an augmentation term that extracts extra information from the current status data. They differ in the component that corresponds to $s=1$, due to the presence of right censoring. In the meantime, this term is similar to the one in $\tau_P^{\text{rc only}}$, except that the indicator function $\mathbbm{1}(u>t^*)$ is now replaced by the function $h^*(u,w)$. While it would thus seem natural that $\tau_P$, which leverages both data sources, must have lower variance than $\tau_P^{\text{rc only}}$, this turns out to not necessarily be the case.  This disappointing phenomenon can arise because the gradient $\tau_P$ in Lemma~\ref{lemma: IF RC and CS} is not necessarily the canonical gradient of $\Phi_{t^*}$ relative to the model $\mathcal{M}$. The exchangeability assumption of the event time distribution on distributions in $\mathcal{M}$ imposes additional constraints on functions in the corresponding tangent space. In particular, scores in the tangent space must arise from the same perturbation to the event time distribution in both the right censored and the current status data. As a result, $\tau_P$ is generally not in the tangent space, except when $\Gamma(\cdot|w)$ is the constant function with value 1 for almost every $w$, that is, when there is in fact no censoring for $S=1$.

To obtain the canonical gradient of $\Phi_{t^*}$, we project the gradient $\tau_P\rconly$ onto the tangent space with respect to $\mathcal{M}$ at $P$ \citep{van2003unified}, which we characterize in the following lemma.

\begin{lemma}\label{lemma: EIF RC and CS}
The canonical gradient of $\Phi_{t^*}: \calM \rightarrow \mathbb{R}$ at $P$ is
\begin{equation}
     \tau_P\eff: x \mapsto s \left\{\delta_R \eta^*(y,w) - \int_0^y \eta^*(t,w)\lambda_{T|W}(t|w) dt\right\} \nonumber + (1-s)\left\{\frac{\delta_C - F_{T|W}(c|w)}{F_{T|W}(c|w)}\right\}\Theta^*(c,w) + \mu(w)-\phi_{t^*},
\end{equation}
where $\eta^*$ is the solution to the following integral equation for almost every $(t,w)$
\begin{equation}\label{eq: integral equation with rc}
    \pi\eta^*(t,w)\Gamma(t|w)S_{T|W}(t|w) + S_{T|W}(t^*|w)\I{t \leq t^*} + (1-\pi)\int_t^\infty \frac{S_{T|W}(c|w)\Theta^*(c,w)}{F_{T|W}(c|w)} g(c|w) dc = 0
\end{equation}
with
\begin{equation}\label{eq: define Theta}
    \Theta^*(c,w) = \int_0^c \eta^*(u,w)\lambda_{T|W}(u|w) du. \nonumber 
\end{equation}
\end{lemma}

Compare to the results in Lemma~\ref{lemma: EIF FO and CS}, the canonical gradient $\tau_P^{\text{eff}}$ relative to $\mathcal{M}$ and its associated integral equation are now presented under a perturbation to the conditional hazard $\lambda_{T|W}$, denoted by $\eta^*$. This perturbation framework allows us to decompose the tangent space into orthogonal subspaces. Notably, unlike in Lemma~\ref{lemma: IF RC and CS}, where the distribution of the censoring time $\Gamma$ only acts on the right-censored data, $\Gamma$ now enters directly into the integral equation \eqref{eq: integral equation with rc}. Thus, the probability $\Gamma(t|w)$, which can be regarded as the probability of no missingness in the right-censored data, also affects the relative weighting of the terms related to the right-censored and the current status data at given value of $(t,w)$ in the integral equation. Moreover, as we have noted previously, $\Gamma(t|w) \equiv 1$ corresponds to the case without censoring for $S=1$. Therefore, setting $\Gamma(t|w)\equiv 1$ in Lemma~\ref{lemma: EIF RC and CS} recovers results presented in Lemma~\ref{lemma: EIF FO and CS}.

\section{One-step estimation and theoretical properties}\label{sec: one-step estimation}
We now use the derived gradients to construct estimators and develop the corresponding inference procedure. Let $\{X_i = (S_i, W_i,S_iY_i,S_i\Delta_{R,i},(1-S_i)\Delta_{C,i},(1-S_i)C_i), i=1,\ldots,n\}$ be an i.i.d. sample drawn from $P$. Let $n_1 = \sum_i S_i$ and $n_0 = n-n_1$. We construct an estimator using the one-step estimation strategy \citep[see, for example, ][]{bickel1993efficient}. Given an initial estimate $\widehat{P}$ of $P$, a one-step estimator takes the form of $\widetilde{\phi}_{t^*} = \Phi_{t^*}(\widehat{P}) + \frac{1}{n}\sum_{i=1}^n\tilde\tau_{\widehat{P}}(X_i)$, where $\Phi_{t^*}(\widehat{P})$ is the plug-in estimator and $\tilde\tau_{\widehat{P}}$ is a gradient of $\Phi_{t^*}$ evaluated at $\widehat P$. Under appropriate conditions, this estimator will be consistent and asymptotically normal with asymptotic variance given by the variance of $\tilde\tau_P$, thus facilitating the construction of confidence intervals.

Notably, the estimate $\widehat{P}$ must include
estimates of all components of $P$ necessary for the evaluation of $\tau$. Under our setting, we need estimates for $F_{T|W}$, $\lambda_{T|W}$, $\Lambda_{T|W}$, $\Gamma$, and $g$ to evaluate $\tau$, and we denote these estimates as $\widehat F_{T|W}$, $\widehat\lambda_{T|W}$, $\widehat\Lambda_{T|W}$,  $\widehat \Gamma$, and $\widehat g$, respectively. Furthermore, we estimate the marginal distribution of the covariate $W$ with its corresponding empirical distribution. Given these nuisance estimates, we can construct the following one-step estimator based on $\tau_P$:
\begin{multline*}
    \widehat{\phi}_{t^*} = \frac{1}{n}\sum_{i=1}^n \Bigg[ \widehat\mu(W_i) + \frac{(1-S_i)(\Delta_{C,i}-\widehat F_{T|W}(C_i|W_i))}{\widehat F_{T|W}(C_i|W_i)(1-\widehat F_{T|W}(C_i|W_i))}\int_0^{C_i} \widehat h^*(u,W_i)d\widehat F_{T|W}(u|W_i) + \\
    S_i \int_0^\infty \frac{\widehat h^*(u,W_i) - \widehat\EE_{T|W}[\widehat h^*(T,W)|T\geq u,W_i]}{\widehat\Gamma(u|W_i)}\left\{d\I{Y_i \leq u,\Delta_{R,i} =1} -\I{Y_i \geq u} d\widehat\Lambda_{T|W}(u|W_i)\right\}\Bigg],
\end{multline*}
where $\widehat\mu(w) = 1-\widehat F_{T|W}(t^*|w)$ and $\widehat{\Lambda}_{T\mid W} = -\log (1-\widehat F_{T|W})$. The estimate $\widehat\EE_{T|W}[\widehat h^*(T,W)|T\geq u,W]$ can be obtained by noting that, for a generic $\widetilde h$, we have $\EE_{T|W}[\widetilde h(T,W)|T\geq u,W] = \{P_{T|W}(T \geq u |W)\}^{-1} \EE_{T|W}[\widetilde h(T,W)I\{T\geq u\}|W].$ In addition, $\widehat{h}^*$ is obtained as the solution to the empirical version of the multi-dimensional Fredholm equation of the second kind in \eqref{eq: integral equation}, where all unknown nuisance functions in $P$ are replaced by their corresponding estimates under $\widehat{P}$. Similarly, we can also construct a one-step estimator based on $\tau_P\eff$:
\begin{multline*}
    \widehat\phi_{t^*}\eff = \frac{1}{n}\sum_{i=1}^n \Bigg[S_i \left\{\Delta_{R,i} \widehat\eta^*(Y_i,W_i) - \int_0^{Y_i} \widehat\eta^*(t,W_i)\widehat\lambda_{T|W}(t|W_i) dt\right\}
    + \\
    (1-S_i)\left\{\frac{\Delta_{C,i} - \widehat F_{T|W}(C_i|W_i)}{\widehat F_{T|W}(C_i|W_i)}\right\}\widehat\Theta^*(C_i,W_i) + \widehat \mu(W_i)\Bigg],
\end{multline*}
where $\widehat\eta^*$ and $\widehat\Theta^*$ are again obtained by solving the integral equation \eqref{eq: integral equation with rc} with all unknown nuisance functions replaced by their corresponding estimates.

Although the relevant Fredholm equations do not have closed-form solutions in general, various numerical approaches are available for solving such integral equations, including kernel approximation \citep{atkinson1976survey,kagiwada1978imbedding}, projection methods \citep{golberg1979solution}, and quadrature methods \citep{prenter1981numerical}. Specifically, we note that \eqref{eq: integral equation} (or \eqref{eq: integral equation with rc}) must be satisfied for (almost) all values of $w$. Thus, for the purpose of implementing the one-step estimator, we can solve one equation for each observed covariate value $w_i$, using these aforementioned approaches. In Appendix~\ref{app: solve equation}, we provide one practical approach based on a discrete approximation, which involves solving a system of linear equations for each covariate value. Alternatively, representing $h^*(t, w_i)$ or $\widehat h^*(t, w_i)$ as a linear combination of basis functions in $t$ reduces the problem to solving for the coefficients of the chosen bases at each covariate value. This allows for smooth evaluation of $h^*(t, w_i)$ at arbitrary $t$ and reduces the computation cost. To further expedite computation, we can consider approximating the solution $h^*$ or $\widehat h^*$ using a basis expansion with, for example, tensor-product bases in $t$ and $w$. Subsequently, we can solve for the coefficients of these basis functions by minimizing the discrepancy between the left-hand and right-hand sides of \eqref{eq: integral equation} over different values of time and covariate in one optimization problem. Similar techniques can be used to solve for $\eta^*$ or $\widehat\eta^*$.

In the following, we will formally establish the asymptotic properties of our proposed estimators under suitable requirements on the nuisance function estimates. First, the following theorem characterizes the consistency of $\widehat{\phi}_{t^*}$. 

\begin{theorem}[Double robustness of $\widehat{\phi}_{t^*}$] \label{thm:double robust phi}
Let $\lambda_{R\mid W}$ denote the conditional hazard function of the censoring time $R$ given $W$, the estimated support of inspection time $\widehat{\mathcal{C}} \in [c_l,c_u]$, and $\|\cdot \|_{L_2(P_W)}$ denote the $L_2(P_W)$-norm.
Under Assumptions~\ref{cond:exchangeability} to \ref{cond:uninformative censoring}, $\widehat{\phi}_{t^*}\xrightarrow[]{p} \phi_{t^*}$ if either (a) $\mathrm{sup}_{c\in [c_l,c_u]}\|\widehat g(c\mid \cdot) - g(c\mid \cdot)\|_{L_2(P_W)} \xrightarrow[]{p} 0$ and $\mathrm{sup}_{c\in [c_l,c_u]}\|\widehat \lambda_{R\mid W}(c\mid \cdot) - \lambda_{R\mid W}(c\mid \cdot) \|_{L_2(P_W)} \xrightarrow[]{p} 0$ or (b) $\mathrm{sup}_{u\leq \max(c_u,t^*)}\|\widehat{F}_{T\mid W}(u\mid \cdot) - F_{T\mid W}(u\mid \cdot) \|_{L_2(P_W)} \xrightarrow[]{p} 0$ and $\mathrm{sup}_{u\leq \max(c_u,t^*)}\|\widehat{\EE}_{T\mid W}[\widehat{h}^*(T,W)\mid T \geq u, W=\cdot] - \EE_{T\mid W}[\widehat{h}^*(T,W)\mid T \geq u, W=\cdot]\|_{L_2(P_W)} \xrightarrow[]{p} 0$.
\end{theorem}

Although the estimator $\widehat{\phi}_{t^*}$ generally has a larger variance compared to $\widehat{\phi}^{\text{eff}}_{t^*}$, it offers the appealing advantage of double robustness. Indeed, Theorem~\ref{thm:double robust phi} suggests that the estimator $\widehat\phi_{t^*}$ is doubly robust in the following sense: if either the conditional distribution of inspection time $C$ given $W$ and that of censoring time $R$ given $W$ are consistently estimated, or the conditional distribution of $T$ given $W$ is consistently estimated, $\widehat\phi_{t^*}$ is consistent. Notably, the conditional expectation $\widehat{\EE}[\widehat{h}^*(T,W)\mid T \geq u, W]$ can be computed under the estimate $\widehat{F}_{T\mid W}$. As a result, condition (b) simply reduces to requiring  $\widehat{F}_{T\mid W}$ being consistent. Meanwhile, as long as $\widehat{h}^*$ solves \eqref{eq: integral equation} under $\widehat{P}$, the difference between $\widehat h^*$ and $h^*$ does not directly appear in the estimation error. Theorem~\ref{thm:double robust phi} implies that, if the relevant nuisances are estimated via parametric or semiparametric models such as Weibull regression or Cox proportional hazards model, the resulting estimator remains consistent if either the models for $C\mid W$ and $R\mid W$ are correctly specified or the model for $T\mid W$ is correctly specified but not necessarily both sets. 

The next theorem shows that when all nuisance functions are consistently estimated with sufficiently fast convergence rate, the one-step estimator is asymptotically linear with influence function $\tau_P$ and thus asymptotically normal.

\begin{theorem}[Asymptotic behavior of $\widehat{\phi}_{t^*}$] \label{thm:efficiency}
Under Assumptions~\ref{cond:exchangeability} to \ref{cond:uninformative censoring}, if $\widehat{g}$, $\widehat{F}_{T\mid W}$, and $\widehat{\lambda}_{R\mid W}$ all belong to a fixed $P$-Donsker class of functions with probability tending to one,  $\Gamma(\max(c_u,t^*) \mid W)>\epsilon$ with probability one for some constant $\epsilon>0$, and 
    \begin{enumerate}
        \item $\mathrm{sup}_{c\in [c_l,c_u]}\|\widehat g(c\mid \cdot) - g(c\mid \cdot)\|_{L_2(P_W)}\lVert \widehat{F}_{T\mid W}(c\mid \cdot)- F_{T\mid W}(c\mid \cdot)\rVert_{L_2(P_W)}  = o_p(n^{-1/2})$; and
        \item $\mathrm{sup}_{u \leq \max(c_u,t^*)}\lVert \widehat \EE_{T\mid W}[\widehat{h}(T,W)\mid T \geq u,W=\cdot]- \EE_{T\mid W}[\widehat{h}(T,W)\mid T \geq u,W=\cdot]\rVert_{L_2(P_W)}  \lVert \widehat{\lambda}_{R\mid W}(u\mid \cdot) - \lambda_{R\mid W}(u\mid \cdot)  \rVert_{L_2(P_W)}  = o_p(n^{-1/2})$.
    \end{enumerate}
     Then, $$\sqrt{n}(\widehat{\phi}_{t^*} - {\phi}_{t^*}) \rightarrow N(0,\text{var}(\tau_{P})).$$
\end{theorem}
The asymptotic normality implies that we can construct a $1-\alpha$ Wald-type confidence interval for $\phi_{t^*}$, $(\widehat{\phi}_{t^*} - n^{-1/2}z_{1-\alpha/2}\widehat\sigma, \widehat{\phi}_{t^*} + n^{-1/2}z_{1-\alpha/2}\widehat\sigma)$ where $\widehat\sigma^2$ is the empirical variance of the estimated gradient $\tau_{\widehat P}$. The estimation errors of the nuisance functions appear in product forms. This suggests that we may use flexible statistical learning tools to consistently estimate the nuisance functions with a sufficiently fast convergence rate, but the convergence rate can be slower than root n. We provide more detailed discussion regarding this point after Theorem~\ref{thm:efficiency real}. 

Similar results can be obtained for the one-step estimator $\widehat\phi_{t^*}\eff$ based on the canonical gradient $\tau_P\eff$ as follows.

\begin{theorem}[Asymptotic behavior of $\widehat{\phi}_{t^*}\eff$] \label{thm:efficiency real}
Under Assumptions~\ref{cond:exchangeability} to \ref{cond:uninformative censoring}, if $\widehat{g}$, $\widehat{\lambda}_{T\mid W}$, and $\widehat{\lambda}_{R\mid W}$ all belong to a fixed $P$-Donsker class of functions with probability tending to one,  $\Gamma(\max(c_u,t^*) \mid W)>\epsilon$ with probability one for some constant $\epsilon>0$, and 
\begin{enumerate}
    \item $\mathrm{sup}_{c\in [c_l,c_u]}\|\widehat g(c\mid \cdot) - g(c\mid \cdot)\|_{L_2(P_W)}  = o_p(n^{-1/4})$; and
    \item $\mathrm{sup}_{u\leq \max(c_u,t^*)} \lVert \widehat{\lambda}_{R\mid W}(u\mid \cdot) - \lambda_{R\mid W}(u\mid \cdot)  \rVert_{L_2(P_W)}  = o_p(n^{-1/4})$; and
    \item $\mathrm{sup}_{u\leq \max(c_u,t^*)}\|\widehat \lambda_{T\mid W}(u\mid \cdot) - \lambda_{T\mid W}(u\mid \cdot)\|_{L_2(P_W)}  = o_p(n^{-1/4})$.
\end{enumerate}
    Then, $$\sqrt{n}(\widehat{\phi}_{t^*}\eff - {\phi}_{t^*}) \rightarrow N(0,\text{var}(\tau_P\eff)).$$
\end{theorem}
The variance of the canonical gradient $\tau_P\eff$ characterizes the semiparametric efficiency bound in estimating the survival probability relative to the model $\calM$. Therefore, the above theorem implies that, $\widehat{\phi}_{t^*}\eff$ will attain this semiparametric efficiency bound if the nuisance function estimates converge fast enough but each can be slower than root-n. As a result, data-adaptive approaches can be used for estimating these nuisance functions. For instance, $\widehat \lambda_{T|W}$ and $\widehat\lambda_{R\mid W}$ may be obtained via fitting (semi)parametric models, kernel methods with inverse probability of censoring weighting or statistical learning methods for survival data, using right-censored data only. The inspection time $C$ is observed for current-status data, and its distribution $g$ can be estimated either parametrically or through nonparametric regression and machine learning methods. Importantly, estimating $\widehat\eta^*$ does not introduce additional asymptotic variance; any error in $\widehat\eta^*$ is entirely attributable to the estimation error in the nuisance functions. Intuitively, this is because the canonical gradient  closely mirrors the structure of the integral equation for $\eta^*$, resulting in cancellation of higher-order errors. Furthermore, the Donsker condition can be removed by employing cross-fitting, thereby permitting the use of flexible statistical learning methods. Given the asymptotic normality, we can form a $1-\alpha$ Wald-type confidence interval for $\phi_{t^*}$, $(\widehat{\phi}_{t^*}\eff - n^{-1/2}z_{1-\alpha/2}\widehat\sigma, \widehat{\phi}_{t^*}\eff + n^{-1/2}z_{1-\alpha/2}\widehat\sigma)$ with $\widehat\sigma^2$ being the empirical variance of the estimated canonical gradient $\tau_{\widehat P}\eff$.

\section{Experiments}\label{sec: experiments}
We conducted a simulation study to evaluate the performance of the proposed data fusion estimators relative to estimators that use only right-censored data or only current status data. We generated a right-censored dataset and a current status dataset, with a sample size ratio of $1{:}2$ and total sample sizes of $n \in \{300, 600, 1500\}$. For both sources, we simulated two covariates: $W_1 \sim \mathrm{Uniform}(0,1)$ and $W_2 \sim \mathrm{Bernoulli}(0.5)$. The event time $T$ was generated from an exponential distribution with rate $0.8 + 0.4 W_1 + 0.2 W_1 W_2$. For the right-censored data, the censoring time was generated from an exponential distribution with rate $1.5 - 0.2 W_1 - 0.5 W_2$ so that it is independent of $T$ given $(W_1,W_2)$. For the current status data, the inspection time $C$ was generated as $C = 0.5 + 0.5  \mathrm{Beta}(1, 0.75 + 0.5 W_1 + 0.1 W_2)$, yielding inspection times $C \in (0.5, 1)$.

We aim to estimate the survival probability at time points $t^* \in \{0.2, 0.7, 0.9\}$. We compared the following estimators: (1) the current status only estimator (CS) proposed by \cite{wolock2025investigating}, (2) the right-censored only AIPCW estimator (RC) constructed via $\tau^{\text{rc only}}_{P}$, (3) the data fusion estimator constructed via $\tau_P$, and (4) the efficient data fusion estimator constructed via $\tau_P^{\text{eff}}$. The CS estimator was constructed using \texttt{survML} R package \citep{survml}, with a library of generalized linear model and generalized additive model under 5-fold cross validation. For the other estimators, the conditional survival and cumulative hazard functions of the event time and censoring time were estimated via \texttt{survSuperLearner} \citep{westling2024inference} with a library consisting of the Kaplain-Meier estimator, Cox proportional hazards model, generalized additive model, and random survival forest, under default settings for all individual learners and \texttt{survSuperLearner}. The conditional density of the inspection times, $g$, is estimated via \texttt{np} R package \citep{hayfield2007np} with default settings. To solve for $h^*$ and $\eta^*$, we solved one equation for each observed covariate value. For given covariate value, we employed a basis expansion in time $t$ using polynomial bases up to degree 10 and their interactions with the indicator $\mathbbm{1}(t > t^*)$, resulting in 22 basis functions in total. All integrals were numerically approximated using a grid of 2000 points. Each simulation scenario was replicated 500 times.

Table~\ref{tab:main_result} displays the main result. When $t^*=0.2$, $\phi_{0.2}$ cannot be identified using current status data alone since it falls out of the support of the inspection time. However, leveraging the current status dataset still contributes to efficiency gain as the proposed data fusion estimator achieves $12\%$ to $13\%$ reduction in CI length compared to the right-censored only estimator (RC). Such efficiency gains are more pronounced at $t^*=0.7$ and $t^*=0.9$, with an overall reduction ranging from $36\%$ to $44\%$. Looking at the results from another perspective, augmenting current status data with right-censored data also yields substantial gains in efficiency, particularly as sample size increases. This pattern is further illustrated in Figure~\ref{fig:rate}. The CS only estimator has a slower convergence rate of $n^{-1/3}$, highlighting the fact that current status data provide substantially less information for estimating survival probabilities compared to right-censored data. When augmenting current status data with right-censored data, however, the proposed approach boosts the convergence rate from cubic root-n to root-n, indicated by the steeper slope. When augmenting right-censored data with current status data, the proposed estimator maintains the root-n rate but with a reduced variance, as indicated by the lower intercept. 

Finally, although developing a method for combining the CS only and RC only estimators with theoretical guarantees falls outside the scope of the current paper, we may consider a naive inverse-variance-weighting approach and use the resulting CI length as a benchmark. For instance, when $N = 1500$ and $t^*=0.9$, the naive inverse-variance weighting approach gives a CI length of 0.100 which is larger than the CI lengths with the data fusion estimators. This is to be expected because the weight assigned to the CS only estimator will approach zero due to its slower convergence rate as sample sizes of the CS and RC data tends to infinity proportionally. As a result, the length of the naive inverse-variance-weighting CI approaches the length of the RC only CI.

\begin{table}
\begin{center}
\caption{Estimated survival probabilities using current status (CS) data only, right-censored (RC) data only, or both. The sample size ratio of CS : RC is 2:1. For the CS-only estimator, only results for $t^*=0.7$ and $t^*=0.9$ are presented since $\phi_{0.2}$ is not identifiable using CS alone.}\label{tab:main_result}
\begin{adjustbox}{width=\textwidth} 
\begin{tabular}{lrrrrrrrrr}
\toprule
   & \multicolumn{3}{c}{$t^* = 0.2$} 
 & \multicolumn{3}{c}{$t^* = 0.7$}  & \multicolumn{3}{c}{$t^* = 0.9$}  \\
\cmidrule(l){2-4} \cmidrule(l){5-7}\cmidrule(l){8-10}
  & Estimate & CI length  &  Coverage    & Estimate & CI length  &  Coverage  & Estimate & CI length  &  Coverage  \\
\midrule
\textbf{N=300} \\
\hspace{2em}CS &  $\cdot$ & $\cdot$ & $\cdot$ & 0.484 &0.197  &0.920  & 0.361 &0.287  & 0.892\\
\hspace{2em}RC &  0.802 & 0.175 & 0.949 &0.468 & 0.291 & 0.929 & 0.380 & 0.301 & 0.931\\
\hspace{2em}Both & 0.805 & 0.153 & 0.935 &  0.475 & 0.167 & 0.943 &  0.387 & 0.183 & 0.945\\
\hspace{2em}Both eff & 0.807 & 0.149 & 0.976 & 0.480 & 0.167 & 0.929 & 0.396 & 0.173 & 0.921\\
\textbf{N=600} \\
\hspace{2em}CS  &  $\cdot$ & $\cdot$ & $\cdot$ & 0.482 & 0.158 & 0.936 & 0.376 &0.221  &0.940 \\
\hspace{2em}RC & 0.809 & 0.124 & 0.950 &  0.473 & 0.200 & 0.952 & 0.384 & 0.212 & 0.948\\
\hspace{2em}Both & 0.812 & 0.109 & 0.946 & 0.479 & 0.112 & 0.956  & 0.390 & 0.131 & 0.944\\
\hspace{2em}Both eff & 0.813 & 0.110 & 0.960 &  0.480 & 0.117 & 0.924  & 0.394 & 0.130 & 0.904\\
\textbf{N=1500} \\
\hspace{2em}CS  & $\cdot$ & $\cdot$ & $\cdot$ & 0.483  & 0.116 & 0.922 & 0.385 & 0.157 & 0.942\\
\hspace{2em}RC & 0.808 & 0.077 & 0.938 &0.477 & 0.122 & 0.960 & 0.386 & 0.130 & 0.958\\
\hspace{2em}Both & 0.810 & 0.067 & 0.956 & 0.481 & 0.071 & 0.954 & 0.391 & 0.083 & 0.942\\
\hspace{2em}Both eff & 0.811 & 0.067 & 0.964 &  0.483 & 0.069 & 0.932 &  0.394 & 0.078 & 0.928\\
\bottomrule
\end{tabular}
\end{adjustbox}
\end{center}
\end{table}

\begin{figure}
    \centering
    \includegraphics[width=0.6\linewidth]{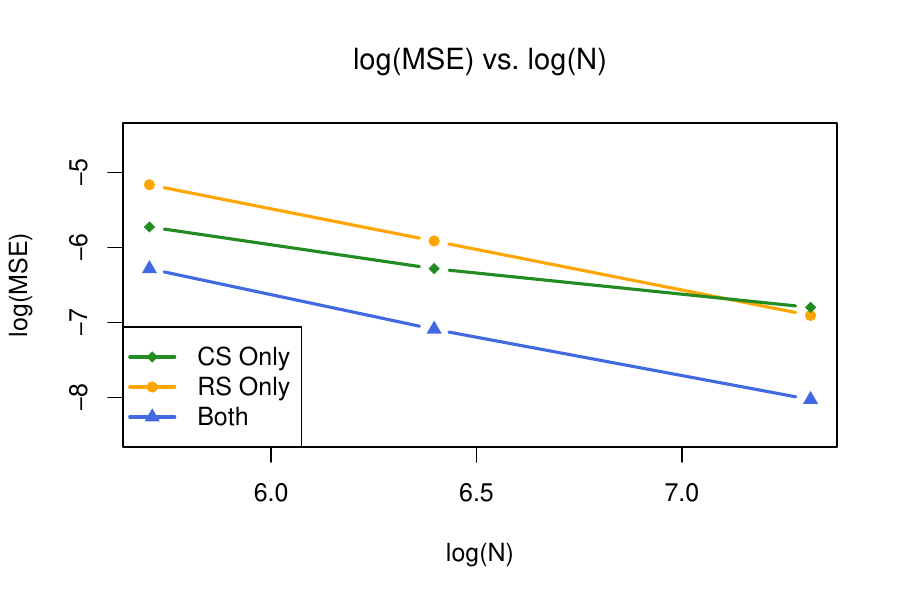}
    \caption{$\log$(MSE) vs. $\log$(N) when $t^*=0.7$ for current status only estimator (CS), right censored only AIPCW estimator (RS), and the proposed data fusion estimator (Both).}
    \label{fig:rate}
\end{figure}

\section{Discussion}\label{sec: discussion}
In this work, we consider integrating right-censored and current status data to efficiently estimate the survival function. Under an exchangeability assumption, we derive the gradients of the survival probability with respect to the observed data fusion model. Subsequently, we develop doubly robust and semiparametrically efficient one-step estimators and corresponding inferential procedures. 

Our exchangeability condition, Assumption~\ref{cond:exchangeability}, requires that both the conditional distribution of the event time given covariates and the marginal distribution of the covariates are the same in the right-censored and current status data. While the exchangeability of the conditional distribution $T \mid W$ is crucial for fusing these two sets of data, the assumption on the covariate distribution can be relaxed. In fact, the covariate distributions are likely to be different when the two datasets are collected from different populations. They may also differ when regular monitoring during study follow-up can only be achieved for a subset of individuals, so that right-censored data are available for this subset and only current status at some random inspection time are available for the rest of the subjects. We refer to the situation where $P_{W|S=1}$ and $P_{W|S=0}$ differ as one with covariate shift. In Appendix~\ref{app: extensions}, we extend our proposed framework to handle such covariate shift and develop efficient estimators for the survival probability in a target population, which can be defined by either the $S=1$ or $S=0$ cohort. Moreover, while we consider two datasets in this paper, our proposed methods can be easily extended to accommodate multiple datasets under suitable exchangeability assumptions and a well-defined target population.

We have focused on inferring the marginal survival probability in observational settings. However, our framework can be further extended to compare the survival probabilities between subgroups defined by some discrete covariate such as the treatment group. Under additional assumptions such as no unmeasured confounding \citep[see, for example, ][]{baer2025causal}, this can correspond to a causal parameter measuring the average treatment effect in terms of the risk difference between the treatment and control groups. In observational studies, such an extension will involve modeling the treatment assignment mechanism and incorporating it into our current framework. We leave this for future investigation.

We have primarily considered statistical inference about the survival probability at some pre-specified time point. While this estimand is of interest in many practical applications, other summaries of the distribution of the survival time can also be considered. One example is the restricted mean survival time (RMST). The RMST can be estimated at root-n rate using current status data alone, and asymptotically normal estimators have been proposed \citep{van1997estimation,van1998locally}. Thus, a standard meta-analysis approach such as inverse variance weighting is directly applicable. However, it remains to be seen whether adapting the current data fusion framework can provide additional efficiency gain under the exchangeability of the conditional distribution of $T|W$.

Both right-censored data and current status data can be regarded as coarsening of the time to event outcome. To the best of our knowledge, existing works on data fusion with coarsened data have focused on fusing data with the same type of coarsening, for example, all time-to-event outcomes subject to right censoring. To this end, our study represents an important step in studying data fusion with mixed types of coarsening in a semiparametric framework. It is worthwhile in the future to study fusing right-censored data with other types of coarsened or incomplete data, including left truncation which falls outside of the coarsening-at-random framework but about which significant progress \citep{wang2024doubly,wang2024learning} has been made recently.

\bibliographystyle{abbrvnat}
\bibliography{reference}

\appendix

\section{Solving the integral equation}\label{app: solve equation}
In this section, we examine more closely the integral equation involved in the gradient of the survival function. We focus on the equation involving $h^*$ and rewrite it into a Fredholm equation of the second kind and propose some practical strategies to obtain a solution. Similar strategies are applicable when we solve the equation involving $\eta^*$. 

Recall that we define the function $h^*(t,w)$ as the solution to the following equation: 
\begin{equation*}
    \pi h^*(t,w)  + (1-\pi)  \int_t^\infty \frac{H^*(c,w) g(c|w)}{F_{T|W}(c|w)(1-F_{T|W}(c|w))} dc - \gamma(w) - I\{t>t^*\} + \mu(w) = 0, \label{eq:solve_h_with_w}
\end{equation*}
where 
\begin{align*}
    H^*(c,w) &= \int_0^c f_{T|W}(t|w)h^*(t,w)dt; \\
    \gamma(w) &= (1-\pi)  \int \int_t^\infty \frac{H^*(c,w) g(c|w)}{F_{T|W}(c|w)(1-F_{T|W}(c|w))} f_{T|W}(t|w)dcdt.
\end{align*}
For the ease of notation, let $Q(c,w): = g(c\mid w)/\{F_{T\mid W}(c|w) (1 - F_{T\mid W}(c|w))\}$. We can write out Equation~\eqref{eq:solve_h_with_w} as
\begin{align*}
    0 & = \pi h^*(t,w) - \mathbbm{1}(t >t^*) + \mu(w) \\
    & \quad +(1-\pi)\int_t^\infty Q(c,w)\int_0^c f_{T\mid W}(s\mid w) h^*(s,w)ds dc \\
    & \quad - (1-\pi) \int \int_t^\infty Q(c,w)\int_0^c f_{T\mid W}(s\mid w) h^*(s,w)ds dc f_{T\mid W}(t\mid w) dt
\end{align*}
Changing the order of variables, we have
\begin{align*}
    0 & = \pi h^*(t,w) - \mathbbm{1}(t >t^*) + \mu(w)\\
    & \quad +(1-\pi)\int f_{T\mid W}(s\mid w) h^*(s,w) \int_{max(t,s)}^\infty Q(c,w)dc ds  \\
    & \quad - (1-\pi) \int f_{T\mid W}(s\mid w) h^*(s,w) \int_s^\infty F_{T\mid W}(c\mid w)Q(c,w)dcds \\
     & = \pi h^*(t,w) - \mathbbm{1}(t >t^*) + \mu(w)\\
    & \quad +(1-\pi)\int f_{T\mid W}(s\mid w) h^*(s,w) \int_{s}^\infty Q(c,w) \mathbbm{1}(c\geq t)dc ds  \\
    & \quad - (1-\pi) \int f_{T\mid W}(s\mid w) h^*(s,w) \int_s^\infty F_{T\mid W}(c\mid w)Q(c,w)dcds \\
    & = \pi h^*(t,w)- \mathbbm{1}(t >t^*) + \mu(w) \\
    & \quad +(1-\pi)\int f_{T\mid W}(s\mid w) h^*(s,w) \int_{s}^\infty Q(c,w) \left(\mathbbm{1}(c\geq t)  - F_{T\mid W}(c\mid w)\right) dc ds.
\end{align*}
Rearranging the above, we have
\begin{equation*}
    h^*(t,w) = \frac{\mathbbm{1}(t>t^*) - \mu(w)}{\pi} + \int h^*(s,w)K(t,s\mid w)ds,
\end{equation*}
where 
\begin{align*}
    K(t,s\mid w) & = - \frac{1-\pi}{\pi}f_{T\mid W}(s\mid w)\int_s^\infty Q(c,w)\left(\mathbbm{1}(c\geq t)  - F_{T\mid W}(c\mid w)\right) dc \\
    &= - \frac{1-\pi}{\pi}f_{T\mid W}(s\mid w)\int_0^\infty \frac{\mathbbm{1}(c>s) g(c\mid w)}{F_{T\mid W}(c|w) (1 - F_{T\mid W}(c|w))}\left(\mathbbm{1}(c\geq t)  - F_{T\mid W}(c\mid w)\right) dc \\
    &= \frac{1-\pi}{\pi}f_{T\mid W}(s\mid w) \left\{ \int_0^t  \frac{\mathbbm{1}(c>s) g(c\mid w)}{ (1 - F_{T\mid W}(c|w))}dc  - \int_t^\infty \frac{\mathbbm{1}(c>s) g(c\mid w)}{F_{T\mid W}(c|w))} dc\right\}.
\end{align*}
We see that the integral limit of $\int h^*(s,w)K(t,s\mid w)ds$ does not involve $t$, and therefore this is a multi-dimensional Fredholm equation of second kind (since $w$ can be multi-dimensional). Note that any such solution will automatically satisfy $$E_P[h^*(T,W)\mid W] = 0.$$ To see this,  we integrate both sides with respect to $f_{T\mid W}$ over $t$ and obtain
\begin{align*}
    \int h^*(t,w)f_{T\mid W}(t\mid w)dt & = \int \left[\int h^*(s,w)K(t,s \mid w)ds\right]f_{T\mid W}(t\mid w)dt\\
    & = \int h^*(s,w)\left[\int K(t,s \mid w)f_{T\mid W}(t\mid w)dt\right]ds,
\end{align*}
where we change the order of integration, and the inner integral can be explicitly written out as
\begin{align*}
    \int K(t,s \mid w)f_{T\mid W}(t\mid w)dt =   \frac{\pi-1}{\pi}f_{T\mid W}(s\mid w) \int_s^\infty Q(c,w) \left[\int \left(\mathbbm{1}(c \geq t) - F_{T\mid W}(c \mid w)\right)f_{T\mid W}(t\mid w)dt\right]dc.
\end{align*}
Computing the bracket, we get
\begin{align*}
    \int \left(\mathbbm{1}(c \geq t) - F_{T\mid W}(c \mid w)\right)f_{T\mid W}(t\mid w)dt &= F_{T\mid W}(c\mid w) - F_{T\mid W}(c\mid w) =0.
\end{align*}

To solve this integral equation, we can solve it for any given value $w$ of $W$. Specifically, we may consider using a dense grid of points $\{u_j: j=1,\ldots,B\}$ for some large number $B$ that adequately covers the support of $T$ and $C$. Then, we can solve for $h^*(u_j,w)$ and then interpolate between the grid point using, for example, nearest neighbor. Specifically, we approximate all integrals with discrete sum and get
\begin{equation*}
    L_j(w):= \pi h^*(u_j,w)  + (1-\pi)  \sum_{u_k > t} \frac{H^*(u_k,w) \{G(u_{k}|w) - G(u_{k-1}|w)\} }{F_{T|W}(u_k|w)(1-F_{T|W}(u_k|w))} - \gamma(w) - I\{u_j>t^*\} + \mu(w) = 0, 
\end{equation*}
with 
\begin{align*}
    H^*(u_k,w) &= \sum_{u_j < u_k} h^*(u_j,w)\{F_{T|W}(u_j|w) - F_{T|W}(u_{j-1}|w)\}; \\
    \gamma(w) &= (1-\pi) \sum_{u_j} \sum_{u_k > u_j} \frac{H^*(u_k,w) \{G(u_{k}|w) - G(u_{k-1}|w)\} }{F_{T|W}(u_k|w)(1-F_{T|W}(u_k|w))}\{F_{T|W}(u_j|w) - F_{T|W}(u_{j-1}|w)\}.
\end{align*}
This defines a system of linear equations in $h^*(u_j,w)$ for each value $w$.

Solving the above linear system for every $w$ value can be computationally intensive, especially if we use a very dense grid to reduce the approximation error and there are many distinct values $w$. Alternatively, for each $w$ value, we may consider a basis approximation of $h^*(t,w)$. Note that the solution $h^*$ must be discontinuous at $t^*$. To adequately capture this discontinuity, we use basis functions interacted with the indicator $I\{t>t^*\}$. Specifically, let $\{b_j(t):j=1,\ldots,\infty\}$ be a set of basis functions in $t$. We may approximate $h^*(t,w)$ with $\sum_{j=1}^J \alpha_j(w)b_j(t)I\{t>t^*\}+\sum_{j=1}^J \gamma_j(w)b_j(t)I\{t\leq t^*\}$ and solve for $\alpha_j(w)$ and $\gamma_j(w)$ by minimizing $\sum_j L_j^2(w)$. This is done for each value $w.$

\sloppy To expedite the computation even further, we can consider the following approximation with a set of tensor-product bases, $h^*(t,w) \approx \sum_{k=1}^K \sum_{j=1}^J \alpha_{jk}\beta_k(w)b_j(t)I\{t>t^*\}+\sum_{k=1}^K \sum_{j=1}^J \gamma_{jk}\beta_k(w)b_j(t)I\{t\leq t^*\}$, where $\{\beta_k(w):k=1,\ldots,\infty\}$ is a set of bases in $w$. We can then solve for the coefficients $\alpha$ and $\gamma$ by minimizing $\sum_k \sum_j L_j^2(w_k)$ with the grid in $t$ and a set of covariate values $\{w_k, 1\leq k \leq B_W\}$. In practice, the set of covariate values can be chosen as the set of observed covariate values.

\section{Inference under covariate shift}\label{app: extensions}
Our exchangeability condition in Assumption~\ref{cond:exchangeability} is two-folded: it imposes (1) that the conditional distributions of the event time given covariates are identical between the right-censored data and the current status data and (2) that the marginal distributions of the covariates are also the same between the two datasets. While the first requirement is often crucial for successfully integrating both sets of data for more efficient estimation and inference, the second requirement can be restrictive. In fact, the study populations from which the two sets of data are collected may differ, for example, if the current status data arise from an observation setting and the right-censored data are collected through careful longitudinal follow-up of a defined cohort. In this case, our second requirement in Assumption~\ref{cond:exchangeability} may be violated.

In this section, we relax the exchangeability condition to the conditional exchangeability condition defined below.
\begin{assumption}[Conditional exchangeability:]\label{cond:conditional exchangeability}
    $T \perp S \mid W$. 
\end{assumption}
\noindent Specifically, this condition assumes that the conditional distributions of the event time $T$ given covariate $W$ are the same for $S=1$ and $S=0$ but places no restrictions on the covariate distributions in the two datasets. Under Assumptions~\ref{cond:uninformative inspection} to \ref{cond:conditional exchangeability}, the distribution of the observed data unit $X = (S,W,SY,S\Delta_R,(1-S)C,(1-S)\Delta_C)$ is uniquely determined by $\Pi$, $P_{T|W}$, $P_{C|W}$ and $P_{R|W}$ introduced in Sections~\ref{sec: fully observed and current status} and \ref{sec: right censored and current status}, as well as two covariate distributions $P_{0,W}$ and $P_{1,W}$. Here, $P_{0,W}$ and $P_{1,W}$ denote the conditional distributions of $W$ given $S=0$ and $S=1$, respectively. Compactly, $P = \mathbb{Q}_2(\Pi, P_{0,W}, P_{1,W} P_{T|W}, P_{C|W}, P_{R|W})$ for some functional $\mathbb{Q}_2$. Two estimands can be of interest, $\Phi_0(P) = \EE_{P_{0,W}}[P_{T|W}(T>t^*|W)]$ and $\Phi_1(P) = \EE_{P_{1,W}}[P_{T|W}(T>t^*|W)]$. The canonical gradients of $\Phi_0$ and $\Phi_1$ can be derived, based on which corresponding efficient estimators can be developed. We provide more details below but first introduce an additional technical assumption.
\begin{assumption}[Overlap]\label{cond:overlap}
    There exists some constant $c_0$ such that $c_0^{-1} \leq dP_{1,W}(w)/dP_{0,W}(w) \leq c_0$, where $dP_{i,W}(\cdot)$ denotes the Radon-Nikodyn derivative of $P_{i,W}$ with respect to some common dominating measure for $i \in \{0,1\}$. 
\end{assumption}
\noindent This condition assumes overlap between the covariate distributions under $S=0$ and $S=1$. For more discussion on this condition, we refer to the readers to \citet{li2023efficient} where the same condition is assumed.

Define the model $\calM_2 = \{\mathbb{Q}_2(\Pi, \widetilde P_{0,W}, \widetilde P_{1,W}, \widetilde P_{T|W}, \widetilde P_{C|W}, \widetilde P_{R|W})\}$ for known $\Pi$ and $(\widetilde P_{0,W}, \widetilde P_{1,W})$ satisfying Assumption~\ref{cond:overlap} but otherwise all distributions unrestricted. For simplicity, let $\phi_i = \Phi_i(P)$. The canonical gradients of $\Phi_0: \calM_2 \rightarrow \mathbb{R}$ and $\Phi_1: \calM_2 \rightarrow \mathbb{R}$ evaluated at $P$ with respective to the model $\calM_2$ are as follows.
\begin{lemma}\label{lemma: EIF covariate shift}
For $i \in \{0,1\}$, a valid gradient of $\Phi_i: \calM_2 \rightarrow \mathbb{R}$ with respect to model $\calM_2$ evaluated at $P$ is given by
\begin{align}\label{eq: IF covariate shift}
    \xi_{i,P}: x &\mapsto s\int_0^\infty \frac{h^*(u,w) - \EE_{T|W}[h^*(T,W)|T\geq u,W=w]}{\Gamma(u|w)}d\left\{\I{y \leq u,\delta_R =1} -\int_0^u \I{y \geq v} d\Lambda_{T|W}(v|w)\right\} \nonumber \\
    &\quad + \frac{(1-s)(\delta_C-F_{T|W}(c|w))}{F_{T|W}(c|w)(1-F_{T|W}(c|w))}\int_0^c f_{T|W}(u|w)h^*(u,w)du + \frac{\mathbbm{1}(s=i)}{\Pi(S=i)}\left\{\mu(w) - \phi_i\right\},
\end{align}
where $h^*$ is the unique solution to
\begin{equation}\label{eq: integral equation covariate shift}
    \pi \frac{dP_{1,W}(w)}{dP_{i,W}(w)} h^*(t,w)  + \frac{dP_{0,W}(w)}{dP_{i,W}(w)}\left\{ (1-\pi)  \int_t^\infty \frac{H^*(c,w) g(c|w)}{F_{T|W}(c|w)(1-F_{T|W}(c|w))} dc - \gamma(w)\right\} - I\{t>t^*\} + \mu(w) = 0,
\end{equation}
with $H^*$ and $\gamma$ defined in the same way as in Lemma~\ref{lemma: EIF FO and CS}.


\end{lemma}

Compared to the gradient we derived in Lemma~\ref{lemma: IF RC and CS}, there are two key differences. First, the integral equation that defines $h^*$ now includes the density ratios of $W$, $dP_{1,W}(w)/dP_{i,W}(w)$ and $dP_{0,W}(w)/dP_{i,W}(w)$, one of which is 1. Second, only observations with $S=i$ contribute to the second term $\mu(W) - \phi_i$, which is not surprising given that we marginalize over $P_{i,W}$ in defining the estimand. 

Based on the derived gradient, we propose the following estimator for $\phi_j$, $j \in \{0,1\}$
\begin{multline*}
    \widehat{\phi}_j = \frac{1}{n}\sum_{i=1}^n \Bigg[\frac{\mathbbm{1}(S_i=j)}{\Pi(S=j)}\widehat\mu(W_i)+  \frac{(1-S_i)(\Delta_{C,i}-\widehat F_{T|W}(C_i|W_i))}{\widehat F_{T|W}(C_i|W_i)(1-\widehat F_{T|W}(C_i|W_i))}\int_0^{C_i} \widehat h^*(u,W_i)d\widehat F_{T|W}(u|W_i) +\\
    S_i \int_0^\infty \frac{\widehat h^*(u,W_i) - \widehat\EE_{T|W}[\widehat h^*(T,W)|T\geq u,W=W_i]}{\widehat\Gamma(u|W_i)}\left\{d\I{Y_i \leq u,\Delta_{R,i} =1} -\I{Y_i \geq u} d\widehat\Lambda_{T|W}(u|W_i)\right\}\Bigg],
\end{multline*}
with $\widehat h^*$ solving the equation analogous to \eqref{eq: integral equation covariate shift} but with all unknown nuisance functions replaced by their corresponding estimates. In addition to the nuisance estimates in Section~\ref{sec: one-step estimation}, we need to estimate the density ratio $dP_{0,W}(w)/dP_{1,W}(w)$. This can be done by estimating the two densities separately via, for example, kernel density estimation and then forming the ratio. Alternatively, by applying the Bayes rule, this can be reduced to estimating $P(S=1|W)$, and a variety of flexible statistical learning tools for classification problems can be applied. Wald-type confidence intervals can be used for statistical inference.

The canonical gradient of $\Phi_i$ can also be derived by projecting a valid gradient onto the observed data tangent space under covariate shift. This will involve an integral equation involving $\eta^*$. Here, a valid initial gradient can be obtained by the influence function of the estimator that uses only the right-censored data to estimate the conditional distribution of the event time given covariates and uses the dataset with $S=i$ for the marginal distribution of the covariates.

Finally, one may also define an estimand by marginalizing the conditional survival probability over the covariate distribution in an external target population. Similar estimators can be formed given a sample of covariate vectors from this target distribution. We conjecture that the gradient will take a very similar form, except that (1) only observations from this external covariate sample contribute the term $\mu(W) - \phi$; and (2) in the definition of $h^*$ (or $\eta^*$), all density ratios have the density function of the external covariate distribution as the denominator. 

\section{Proof of lemmas and theorems}\label{app: proof}

\subsection{Proof of lemmas in Section~\ref{sec: fully observed and current status} and Section~\ref{sec: right censored and current status}}
\begin{proof}[Proof of Lemma~\ref{lemma: tangent space FO and CS}]
Recall that we have introduced an intermediate observation unit $X^I = (W,ST,(1-S)\Delta_C, (1-S)C)$ with distribution $P^I$. Under Assumptions~\ref{cond:exchangeability} and \ref{cond:uninformative inspection}, the density function of $P^I$, denoted as $p$, is given by
\begin{align*}
    p(x^I) &= \pi^s (1-\pi)^{1-s}p_W(w)f_{T|W}(t|w)^s \left\{g(c|w)\left(\int_0^c f_{T|W}(t|w)dt\right)^\delta \left(1-\int_0^c f_{T|W}(t|w)dt\right)^{1-\delta} \right\}^{1-s}.
\end{align*}
We recall that $\pi$ is the probability of $S=1$, $p_W$ is the density function of the covariate distribution, $f_{T|W}$ is the conditional density function of $T|W$ and $g$ is the conditional density function of $C|W$. Now consider a perturbed distribution with density function $p_\epsilon$ such that
\begin{align*}
   p_\epsilon(x^I) &= \pi^s (1-\pi)^{1-s}p_W(w)(1+\epsilon h_W(w))\left\{f_{T|W}(t|w)(1+\epsilon h(t,w))\right\}^s \left\{g(c|w)(1+\epsilon h_C(c,w))\right\}^{1-s}\\
   &\quad \times \left\{\left(\int_0^c f_{T|W}(t|w)(1+\epsilon h(t,w))dt\right)^{\delta_C} \left(1-\int_0^c f_{T|W}(t|w)(1+\epsilon h(t,w))dt\right)^{1-\delta_C} \right\}^{1-s},
\end{align*}
for some functions $h_W \in L_2^0(P_W)$, $h(t,w) \in L_2^0(P_{T|W})$ and $h_C(c,w) \in L_2^0(P_{C|W})$. The log density is given by
\begin{align*}
    \log p_\epsilon(x) &= \log p_W(w)+ \log(1+\epsilon h_W(w)) + s\log f(t|w) + s\log (1+ \epsilon h(t,w)) \\
    &\quad + (1-s)\log g(c|w) + (1-s)\log (1+ h_C(c,w)) \\
    &\quad + (1-s)\delta_C \log \int_0^c f(t|w)(1+\epsilon h(t,w))dt \\
    &\quad + (1-s)(1-\delta_C) \log \left\{1-\int_0^c f(t|w)(1+\epsilon h(t,w))dt \right\}.
\end{align*}
Taking derivatives, we get the tangent space of the following form
\begin{align*}
    \calT^I &= \overline{\textnormal{span}}\Bigg\{h_W(w) + sh(t,w) + (1-s)h_C(c,w) + \frac{(1-s)(\delta_C-F_{T|W}(c|w))}{F_{T|W}(c|w)(1-F_{T|W}(c|w))}\int_0^c f(u|w)h(u,w)du, \\
    &\quad \quad \quad h_W \in L_2^0(P_W), h(t,w) \in L_2^0(P_{T|W}), h_C(c,w) \in L_2^0(P_{C|W})\Bigg\}.
\end{align*}
\end{proof}

\begin{proof}[Proof of Lemma~\ref{lemma: EIF FO and CS}]
To derive the canonical gradient of $\Psi$, we first note that the functional $\Psi$ does not depend on $P_{C|W}$. Therefore, for the purpose of deriving the canonical gradient, it suffices to consider a reduced model with $P_{C|W}$ fixed whose tangent space is as follows
\begin{align*}
    \widetilde\calT^I &= \Bigg\{h_W(w) + sh(t,w) + \frac{(1-s)(\delta_C -F_{T|W}(c|w))}{F_{T|W}(c|w)(1-F_{T|W}(c|w))}\int_0^c f(u|w)h(u,w)du, \\
    &\quad \quad \ h_W \in L_2^0(P_W), h(t,w) \in L_2^0(P_{T|W})\Bigg\}.
\end{align*}

The canonical gradient is an element in the tangent space such that its inner product with the score corresponding to any regular parametric submodel through $P^I$ is equal to the derivative of the parameter of interest along that submodel at $P^I$. First, we consider a submodel $\{P_\epsilon^I:\epsilon\}$ that only perturbs the marginal distribution of $W$ such that $p_{\epsilon,W}(w) = p_W(w)(1+\epsilon h_W(w))$. Recall that $\mu(w) = P_{T|W}(t>t^*|W=w)$ and let $\mu = E_{P_W}[\mu(W)] = S(t^*)$. The following holds for all $h_W \in L_2^0(P_W)$ for some $h_W^* \in L_2^0(P_W)$ and some $h^* \in L_2^0(P_{T|W})$:
\begin{align*}
    \frac{d\Psi(P_\epsilon)}{d\epsilon} &= \int \left( I\{t > t^* \} - S(t^*)\right) h_W(w) f(t|w) p_W(w) dtdw \\
    &= \int (\mu(w)-\mu) h_W(w)p_W(w)dw \\
    &= \EE_P\left[h_W(W)\left\{h_W^*(W) + Sh^*(T,W) + \frac{(1-S)(\Delta_C - F_{T|W}(C|W))}{F_{T|W}(C|W)(1-F_{T|W}(C|W))}\int_0^C f(u|W)h^*(u,W)du\right\}\right] \\
    &= E_{P_W}\left[h_W(W)h_W^*(W)\right] \\
    &= \int h_W(w)h_W^*(w)p_W(w)dw,
\end{align*}
where from the third to the fourth line, we used the fact that $\EE_{P_{T|W,S=1}}[h^*(T,W)|W=w] = 0$ for all $w$, and that conditional on $S=0$ and $(C,W)$, $\Delta_C$ has mean $F_{T|W}(C|W)$. As the above holds for all $h_W$, we must have
\begin{equation*}
    h_W^*(w) = \mu(w) - \mu.
\end{equation*}

Next, we consider parametric submodels that perturb $P_{T|W}$ and find $h^*(t,w)$. Specifically, we now consider a submodel $\{P_\epsilon^I:\epsilon\}$ such that the conditional density of $T$ given $W$ under $P_\epsilon^I$ is $f(t|w)(1+\epsilon h(t,w))$. For the ease of notation, define the function $H^*(c,w) = \int_0^c f(t|w)h^*(t,w)dt$.
\begin{align*}
    \frac{d\Psi(P_\epsilon)}{d\epsilon} &= \int \left( I\{t>t^*\} - \mu(w)\right) h(t,w) f(t|w)  p_W(w) dtdw \\
    &= \EE_P\Bigg[\left\{Sh^*(T,W) + \frac{(1-S)(\Delta_C - F_{T|W}(C|W))}{F_{T|W}(C|W)(1-F_{T|W}(C|W))}\int_0^C f(u|W)h^*(u,W)du\right\} \\
    &\quad \qquad \times \left\{Sh(T,W) + \frac{(1-S)(\Delta_C - F_{T|W}(C|W))}{F_{T|W}(C|W)(1-F_{T|W}(C|W))}\int_0^C f(u|W)h(u,W)du\right\}\Bigg] \\
    &= \EE_P[Sh(T,W)h^*(T,W)] \\
    &\quad + \EE_P\left[(1-S)\frac{(\Delta_C - F_{T|W}(C|W))^2}{F_{T|W}(C|W)^2(1-F_{T|W}(C|W))^2}\int_0^C f(v|W)h(v,W)dv\int_0^C f(u|W)h^*(u,W)du\right] \\
    &= \pi \int h(t,w)h^*(t,w) f(t|w)p_W(w)dtdw \\
    &\quad + (1-\pi)\int \frac{H^*(c,w)}{F_{T|W}(c|w)(1-F_{T|W}(c|w))} \left(\int_0^c f(t|w)h(t,w)dt \right) g(c|w) p_W(w) dcdw \\
    &= \pi \int h(t,w)h^*(t,w) f(t|w)p_W(w)dtdw \\
    &\quad + (1-\pi)\int \left( \int_t^\infty \frac{H^*(c,w) g(c|w)}{F_{T|W}(c|w)(1-F_{T|W}(c|w))} dc \right) h(t,w) f(t|w) p_W(w)\ dtdw \\
    &= \int \left\{ \pi h^*(t,w)  + (1-\pi)  \int_t^\infty \frac{H^*(c,w) g(c|w)}{F_{T|W}(c|w)(1-F_{T|W}(c|w))} dc \right\} h(t,w) f(t|w) p_W(w)\ dtdw
\end{align*}
Hence, the following holds for all $h(t,w)$ such that $\EE_{P_{T|W}}[h(T,W)|W=w] =0$,
\begin{multline*}
    \int \left\{ \pi h^*(t,w)  + (1-\pi)  \int_t^\infty \frac{H^*(c,w) g(c|w)}{F_{T|W}(c|w)(1-F_{T|W}(c|w))} dc - \gamma(w) - I\{t>t^*\} + \mu(w)\right\} \times \\
    h(t,w) f(t|w) p_W(w)\ dtdw = 0,
\end{multline*}
where
\begin{equation*}
    \gamma(w) = (1-\pi)  \int \int_t^\infty \frac{H^*(c,w) g(c|w)}{F_{T|W}(c|w)(1-F_{T|W}(c|w))} f(t|w)dcdt.
\end{equation*}
By construction, the function in the brackets has mean 0 given $W=w$ for all $w$. Therefore, we must have that, for almost every $(t,w)$,
\begin{equation*}
    \pi h^*(t,w)  + (1-\pi)  \int_t^\infty \frac{H^*(c,w) g(c|w)}{F_{T|W}(c|w)(1-F_{T|W}(c|w))} dc - \gamma(w) - I\{t>t^*\} + \mu(w) = 0.
\end{equation*}
The gives an integral equation in $t$ for every given value of $w$. Solving this equation, we obtain $h^*$ and get the canonical gradient
\begin{equation*}
    \tau_{P^I}: x^I \mapsto sh^*(t,w) + \frac{(1-s)(\delta_C - F_{T|W}(c|w))}{F_{T|W}(c|w)(1-F_{T|W}(c|w))}\int_0^c f(u|w)h^*(u,w)du + \mu(w) - \psi.
\end{equation*}
\end{proof}

\begin{proof}[Proof of Lemma~\ref{lemma: IF RC and CS}]
Recall that the observation unit is $X = (S,W,S\Delta_R,SY,(1-S)C,(1-S)\Delta_C)$. That is, for observations with $S=1$, the time-to-event outcome is subject to right censoring. In fact, $X$ can be regarded as a further coarsening of $X^I = (S, W,ST,(1-S)\Delta_C, (1-S)C)$, and this coarsening satisfies the coarsening-at-random (CAR) assumption. To see this, we note that for observations with $S=0$, $X=X^I$ as only current status information is available. Meanwhile, for $S=1$, under the conditional independent censoring in Assumption~\ref{cond:uninformative censoring}, CAR reduces to the CAR in the usual setting of time-to-event outcomes subject to right censoring.  

Therefore, applying the results in Example 1.12 of \citet{van2003unified} and results in \citet{van2007note}, a valid gradient relative to the observed data model is given by
\begin{align*}
    \tau: x &\mapsto s\int_0^\infty \frac{h^*(u,w) - \EE_{T|W}[h^*(T,W)|T\geq u,W=w]}{\Gamma(u|w)}d\left\{I\{y \leq u,\delta_R = 1\} -\int_0^u I\{y \geq v\} d\Lambda_{T|W}(v|w)\right\} \\
    &\quad + \frac{(1-s)(\delta_C - F_{T|W}(c|w))}{F_{T|W}(c|w)(1-F_{T|W}(c|w))}\int_0^c f(t|w)h^*(t,w)dt + \mu(w) - \phi,
\end{align*}
where $\Gamma(u|w) = P_{R|W}(R\geq u|W=w)$ and $\Lambda_{T|W}(t|w)$ is the cumulative hazard function at $t$ given $W=w$. A similar result is available in \citet{tsiatis2006semiparametric} (10.76), although in a different form involving the martingale associated with the censoring process.

Finally, we have the following equality
\begin{align*}
    \EE_{T|W}[h^*(T,W)|T\geq u,W=w] &= \frac{1}{P_{T|W}(T \geq u |W=w)} \EE_{T|W}\left[h^*(T,W)I\{T\geq u\}|W=w\right] \\
    &= \frac{1}{P_{T|W}(T \geq u |W=w)} \EE_{P|S=1}\left[\frac{\Delta_R}{\Gamma(Y|W)}h^*(Y,W)I\{Y\geq u\}|W=w\right].
\end{align*}
\end{proof}

Before proving the results in Lemma~\ref{lemma: EIF RC and CS}, we first present a lemma characterizing the observed data tangent space under the fusion model.
\begin{lemma}\label{lemma: tangent space observed}
The tangent space at $P$ reletive to the observed data model $\calM$ is given by
    \begin{align}
    \calT &= \overline{\textnormal{span}}\Bigg\{s \left\{\delta_R \eta(y,w) - \int_0^y \eta(t,w)\lambda_{T|W}(t|w) dt\right\} \nonumber + (1-s)\left\{\frac{\delta_C - F_{T|W}(c|w)}{F_{T|W}(c|w)}\right\}\int_0^c \eta(t,w)\lambda_{T|W}(t|w) dt \nonumber \\
    &\quad \quad \quad + s \left\{(1-\delta_R) \eta_R(y,w) - \int_0^y \eta_R(t,w)\lambda_{R|W}(t|w) dt\right\} + (1-s)h_C(c,w) + h_W(w), \nonumber \\
    &\quad \quad \quad h_W \in L_2^0(P_W), h_C(c,w) \in L_2^0(P_{C|W})\Bigg\}.
\end{align}
\end{lemma}

\begin{proof}[Proof of Lemma~\ref{lemma: tangent space observed}]
To derive the tangent space at the observed data distribution $P$ with respect to the observed data model, we start with the likelihood function of the observed data. Recall that we use $\Lambda_{R|W}$ and $\Lambda_{T|W}$ to denote the cumulative hazard function corresponding to the conditional distribution of the censoring time and the event time. Moreover, $\Lambda_{R|W}(t|w) = \int_0^t \lambda_{R|W}(u|w)du$ and $\Lambda_{T|W}(t|w) = \int_0^t \lambda_{T|W}(u|w)du$, where $\lambda_{R|W}$ and $\lambda_{T|W}$ are the corresponding hazard functions.

Under Assumptions~\ref{cond:exchangeability} to \ref{cond:uninformative censoring}, the likelihood function is given by
\begin{multline*}
    L(x) = p_W(w)\left\{\pi \Gamma(y|w)S_{T|W}(y|w)\lambda_{R|W}(y|w)^{1-\delta_R}\lambda_{T|W}(y|w)^{\delta_R}\right\}^s \times \\
    \left\{(1-\pi)g(c|w)F_{T|W}(c|w)^{\delta_C}S_{T|W}(c|w)^{1-\delta_C}\right\}^{1-s},
\end{multline*}
with log likelihood function
\begin{align*}
    \log L(x) &= \log p_w(w) + s\left\{\log\pi-\Lambda_{R|W}(y|w)-\Lambda_{T|W}(y|w)+\delta_R\log \lambda_{T|W}(y|w) + (1-\delta_R)\log \lambda_{R|W}(y|w)\right\} \\
    &\quad + (1-s)\left\{\log (1-\pi) + \log g(c|w) - (1-\delta_C)\Lambda_{T|W}(c|w) + \delta_c \log F_{T|W}(c|w)\right\}.
\end{align*}

First, we consider submodels that perturb the conditional distribution $T|W$ such that $\lambda_{\epsilon,T|W}(t|w) = \lambda_{T|W}(t|w)\exp(\epsilon \eta(t,w))$. Such a perturbation leads to the following score
\begin{align*}
    &\quad s \left\{\delta_R \eta(y,w) - \int_0^y \eta(t,w)\lambda_{T|W}(t|w) dt\right\} + (1-s)\left\{\delta_C \frac{S_{T|W}(c|w)}{F_{T|W}(c|w)}-(1-\delta_C)\right\}\int_0^c \eta(t,w)\lambda_{T|W}(t|w) dt \\
    &= s \left\{\delta_R \eta(y,w) - \int_0^y \eta(t,w)\lambda_{T|W}(t|w) dt\right\} + (1-s)\left\{\frac{\delta_C - F_{T|W}(c|w)}{F_{T|W}(c|w)}\right\}\int_0^c \eta(t,w)\lambda_{T|W}(t|w) dt.
\end{align*}
For notational convenience, we define the Martingale $M_T(t) = I\{Y\leq t, \Delta_R = 1\} - \int_0^t I\{Y\geq u\}\lambda_{T|W}(u|W)du$ and let $\tilde M_T(t)$ denote a realized version of it with $(Y,\Delta_R, W)$ replaced by $(y,\delta_R,w)$. Then the above score can be compactly written as
\begin{equation*}
    s\int \eta(t,w)d\tilde{M}_T(t)+ (1-s)\left\{\frac{\delta_C - F_{T|W}(c|w)}{F_{T|W}(c|w)}\right\}\int_0^c \eta(t,w)\lambda_{T|W}(t|w) dt.
\end{equation*}

Similarly, we consider perturbations to the conditional distribution $R|W$ such that $\lambda_{\epsilon,R|W}(t|w) = \lambda_{T|W}(t|w)\exp(\epsilon \eta_R(t,w))$. The corresponding score is as follows
\begin{equation*}
    s \left\{(1-\delta_R) \eta_R(y,w) - \int_0^y \eta_R(t,w)\lambda_{R|W}(t|w) dt\right\} = s\int \eta_R(t,w)d\tilde{M}_R(t),
\end{equation*}
where we define the censoring Martingale $M_R(t) = I\{Y\leq t, \Delta_R = 0\} - \int_0^t I\{Y\geq u\}\lambda_{R|W}(u|W)du$ and the corresponding $\tilde M_R(t)$ in a similar fashion.

The scores corresponding to perturbations to the conditional distribution $C|W$ and the marginal distribution of $W$ are the same as in Section~\ref{sec: fully observed and current status}. Thus, we obtain the observed data tangent space. Moreover, the scores generated by perturbing different component distributions to $P$ are orthogonal. See also Chapter 5 of \citet{tsiatis2006semiparametric}.
\end{proof}

\begin{proof}[Proof of Lemma~\ref{lemma: EIF RC and CS}]
Recall that with right-censored data and current status data, the augmented inverse probability of censoring weighted estimator using the right-censored data only is an RAL estimator, although it is in general not the most efficient one. Therefore, the following function is a valid gradient
\begin{align*}
    \tau_P^{\text{rc only}}:
    x &\mapsto \frac{s}{\pi} \int_0^\infty \frac{-S_{T|W}(t^*|w)I\{u \leq t^*\}}{\Gamma(u|w)S_{T|W}(u_{-}|w)} d\tilde{M}_T(u) + \frac{s}{\pi}(\mu(w)-\phi).
\end{align*}
The canonical gradient can be obtained by projecting $\tau_P^{\text{rc only}}$ onto the observed data tangent space $\calT$. Note that the second term is a function of $s$ and $w$ only, therefore we only need to project it onto the subspace $L_2^0(P_W)$. This projection is given by $\mu(w)-\phi$ due to the independence between $S$ and $W$. We thus focus on the projection of the first term.

Due to the presence of the indicator $I\{s=1\}$ (or equivalently the factor $s$) in this first term and its conditional mean given $W$ being 0, we only need to project it onto the subspace of scores generated by perturbing $R|W$ and $T|W$. Moreover, as the first term is an integral involving the event Martingale, it is already orthogonal to the censoring time score. This implies that the projection takes the following form
\begin{align*}
    &\quad s \left\{\delta_R \eta^*(y,w) - \int_0^y \eta^*(t,w)\lambda_{T|W}(t|w) dt\right\} \nonumber + (1-s)\left\{\frac{\delta_C - F_{T|W}(c|w)}{F_{T|W}(c|w)}\right\}\int_0^c \eta^*(t,w)\lambda_{T|W}(t|w) dt \\
    &= s\int \eta^*(t,w)d\tilde{M}_T(t)+ (1-s)\left\{\frac{\delta_C - F_{T|W}(c|w)}{F_{T|W}(c|w)}\right\}\int_0^c \eta^*(t,w)\lambda_{T|W}(t|w) dt.
\end{align*}
for some function $\eta^*$. Therefore, finding the projection reduces to finding the function $\eta^*$. 

For the convenience of notation, let us define a function $\eta^\dagger(u,w)$
\begin{equation*}
    \eta^\dagger(u,w) = \frac{-S_{T|W}(t^*|w)I\{u \leq t^*\}}{\pi\Gamma(u|w)S_{T|W}(u_{-}|w)}
\end{equation*}

By the definition of projection, we have the following holds for any function $\eta$:
\begin{align*}
    0 &= \EE\Bigg[\left\{S\int (\eta^*(t,W) - \eta^\dagger(t,W))dM_T(t)+ (1-S)\left\{\frac{\Delta_C - F_{T|W}(C|W)}{F_{T|W}(C|W)}\right\}\int_0^C \eta^*(t,W)\lambda_{T|W}(t|W) dt \right\} \times \\
    &\qquad \quad \left\{ S\int \eta(t,W)dM_T(t)+ (1-S)\left\{\frac{\Delta_C - F_{T|W}(C|W)}{F_{T|W}(C|W)}\right\}\int_0^C \eta(t,W)\lambda_{T|W}(t|W) dt\right\}\Bigg] \\
    &= \EE\Bigg[S\int (\eta^*(t,W) - \eta^\dagger(t,W))dM_T(t) \int \eta(t,W)dM_T(t) \Bigg] \\
    &\quad + \EE\Bigg[(1-S)\left\{\frac{\Delta_C - F_{T|W}(C|W)}{F_{T|W}(C|W)}\right\}^2\int_0^C \eta^*(t,W)\lambda_{T|W}(t|W) dt \int_0^C \eta(t,W)\lambda_{T|W}(t|W) dt \Bigg] \\
    &= \pi\EE\Bigg[\int (\eta^*(t,W) - \eta^\dagger(t,W))dM_T(t) \int \eta(t,W)dM_T(t) \Bigg] \\
    &\quad + (1-\pi)\EE\Bigg[\frac{S_{T|W}(C|W)}{F_{T|W}(C|W)}\int_0^C \eta^*(t,W)\lambda_{T|W}(t|W) dt \int_0^C \eta(t,W)\lambda_{T|W}(t|W) dt \Bigg] 
    &\intertext{The covariance of martingale stochastic integrals in the first term can be computed by deriving the expectation of the predictable covariation process \citep{fleming2013counting}. Continuing the above, we have}
    &= \pi\EE\Bigg[\int (\eta^*(t,W) - \eta^\dagger(t,W))\eta(t,W) Y(t)\lambda_{T\mid W}(t\mid W)dt \Bigg] \\
    &\quad + (1-\pi)\EE\Bigg[\frac{S_{T|W}(C|W)}{F_{T|W}(C|W)}\int_0^C \eta^*(t,W)\lambda_{T|W}(t|W) dt \int_0^C \eta(t,W)\lambda_{T|W}(t|W) dt \Bigg] \\
    &= \pi\int \int (\eta^*(t,W) - \eta^\dagger(t,W))\eta(t,w)S_{T|W}(t_{-}|w)\Gamma(t|w)\lambda_{T|W}(t|w)dt\ p(w)dw \\
    &\quad + (1-\pi)\int \int \frac{S_{T|W}(c|w)}{F_{T|W}(c|w)}\int_0^c \eta^*(u,w)\lambda_{T|W}(u|w) du \int_0^c \eta(t,w)\lambda_{T|W}(t|w) dt g(c|w)dc \ p(w)dw \\
    &= \pi\int \int (\eta^*(t,W) - \eta^\dagger(t,W))\eta(t,w)S_{T|W}(t_{-}|w)\Gamma(t|w)\lambda_{T|W}(t|w)dt\ p(w)dw \\
    &\quad + (1-\pi) \int \int \left(\int_t^\infty \frac{S_{T|W}(c|w)}{F_{T|W}(c|w)}\left\{\int_0^c \eta^*(u,w)\lambda_{T|W}(u|w) du\right\}g(c|w) dc\right) \eta(t,w)\lambda_{T|W}(t|w)  dt \ p(w)dw \\
    &= \pi\int \int (\eta^*(t,W) - \eta^\dagger(t,W))\eta(t,w)\Gamma(t|w)f_{T|W}(t|w)dt\ p(w)dw \\
    &\quad + (1-\pi) \int \int \left(\int_t^\infty \frac{S_{T|W}(c|w)}{S_{T|W}(t_{-}|w)F_{T|W}(c|w)}\left\{\int_0^c \eta^*(u,w)\lambda_{T|W}(u|w) du\right\}g(c|w) dc\right) \\
    &\qquad \qquad \qquad \qquad \qquad \qquad \times \eta(t,w)f_{T|W}(t|w)  dt \ p(w)dw \\
    &= \int\int \Bigg\{\pi(\eta^*(t,W) - \eta^\dagger(t,W))\Gamma(t|w) \\
    &\qquad + (1-\pi)\left(\int_t^\infty \frac{S_{T|W}(c|w)}{S_{T|W}(t_{-}|w)F_{T|W}(c|w)}\left\{\int_0^c \eta^*(u,w)\lambda_{T|W}(u|w) du\right\}g(c|w) dc\right)\Bigg\} \\
    &\qquad \qquad \qquad \qquad \qquad \qquad \times \eta(t,w)f_{T|W}(t|w)  dt \ p(w)dw.
\end{align*}
As $\eta$ is arbitrary, we can take it to be the function in the curly brackets and the integral becomes the second moment of this function which needs to be 0. This implies that the function in the curly brackets is 0 for almost every $(t,w)$.

We thus obtain the following equation
\begin{equation}\label{eq: integral equation new}
    \pi(\eta^*(t,w) - \eta^\dagger(t,W))\Gamma(t|w) + (1-\pi)\left(\int_t^\infty \frac{S_{T|W}(c|w)g(c|w)}{S_{T|W}(t_{-}|w)F_{T|W}(c|w)}\left\{\int_0^c \eta^*(u,w)\lambda_{T|W}(u|w) du\right\} dc\right) = 0.
\end{equation}
Substituting the expression for $\eta^\dagger$ back into \eqref{eq: integral equation new}, we get
\begin{equation*}
    \pi\eta^*(t,w)\Gamma(t|w)S_{T|W}(t_{-}|w) + S_{T|W}(t^*|w)I\{t \leq t^*\} + (1-\pi)\int_t^\infty \frac{S_{T|W}(c|w)\Theta^*(c,w)}{F_{T|W}(c|w)} g(c|w) dc = 0.
\end{equation*}
with
\begin{equation*}
    \Theta^*(c,w) = \int_0^c \eta^*(u,w)\lambda_{T|W}(u|w) du.
\end{equation*}
The EIF is therefore
\begin{equation*}
     x\mapsto s \left\{\delta_R \eta^*(y,w) - \int_0^y \eta^*(t,w)\lambda_{T|W}(t|w) dt\right\} \nonumber + (1-s)\left\{\frac{\delta_C - F_{T|W}(c|w)}{F_{T|W}(c|w)}\right\}\Theta^*(c,w) + \mu(w)-\phi.
\end{equation*}
Note that if the distribution of the event time is continuous, we have $S_{T|W}(t_{-}|w) = S_{T|W}(t|w)$.
\end{proof}

Finally, we note that the case where $T$ is fully observed for observations with $S=1$ can be regarded as a special case with $\Gamma(t|w) \equiv 1$ for all $(t,w)$. Then, equation \eqref{eq: integral equation with rc} becomes
\begin{equation}\label{eq: integral equation fo alt}
    \pi\eta^*(t,w)S_{T|W}(t|w) + S_{T|W}(t^*|w)I\{t \leq t^*\} + (1-\pi)\int_t^\infty \frac{S_{T|W}(c|w)\Theta^*(c,w)}{F_{T|W}(c|w)} g(c|w) dc = 0
\end{equation}
with
\begin{equation*}
    \Theta^*(c,w) = \int_0^c \eta^*(u,w)\lambda_{T|W}(u|w) du.
\end{equation*}
The corresponding EIF is therefore
\begin{equation*}
     s \left\{\eta^*(t,w) - \int_0^t \eta^*(u,w)\lambda_{T|W}(u|w) du\right\} \nonumber + (1-s)\left\{\frac{\delta_C - F_{T|W}(c|w)}{F_{T|W}(c|w)}\right\}\Theta^*(c,w) + \mu(w)-\phi.
\end{equation*}
This should result in the same EIF as the one presented in Lemma~\ref{lemma: EIF FO and CS}, and the following lemma shows that this is indeed the case.

\begin{lemma}\label{lemma: equivalent of integral equations}
    Let $h^*$ and $\eta^*$ be the solutions to equations \eqref{eq: integral equation} and \eqref{eq: integral equation fo alt} respectively. Then, the following holds for almost every $(t,w)$:
    \begin{align*}
        \eta^*(t,w) &= h^*(t,w) + \frac{H^*(t,w)}{S_{T|W}(t|w)} \\
        \Theta^*(t,w) &= \frac{H^*(t,w)}{S_{T|W}(t|w)}.
    \end{align*}
\end{lemma}

\subsection{Proof of theorems in Section~\ref{sec: one-step estimation}}
To study the asymptotic properties of our proposed estimator, we will leverage an expansion of a general one-step estimator. Recall that our estimator can be written as $\widehat\phi = \Phi(\widehat P) + P_n \tau_{\widehat P}$, where we use $P_n Q$ to denote the empirical mean of a generic function $Q$. Furthermore, let $PQ$ denote $\EE_P[Q(X)]$ for a generic function $Q$ of $X$.
\begin{align*}
    \widehat\phi - \phi &= \Phi(\widehat P) + P_n \tau_{\widehat P} - \Phi(P) \\
    &= P_n \tau_{\widehat P} - P_n \tau_{P} + P_n \tau_{P} + \Phi(\widehat P) - \Phi(P) \\
    &= (P_n - P)\tau_{P} + (P_n - P) \left\{\tau_{\widehat P} - \tau_{P}\right\} + \{P\tau_{\widehat P} + \Phi(\widehat P) - \Phi(P)\}.
\end{align*}
The first term in the last line of the above display is asymptotically normal by the central limit theorem. The second term is a second-order term, which can be controlled with a Glivenko-Cantelli class or Donsker class theorem. The third term is a remainder term, and we introduce the following notation $\Rem(\widehat P, P) = P\tau_{\widehat P} + \Phi(\widehat P) - \Phi(P)$. Before proving our main theorems, we study $\Rem(\widehat P, P)$ more closely and show that it can be upper bounded using the estimation error of the nuisance functions.
\small
\begin{align*}
    &\Rem(\widehat P, P) \\
    &= \EE_P\left[S\int_0^\infty \frac{ \widehat h(u,W) - \widehat\EE[\widehat h(T,W)|T\geq u,W]}{\widehat\Gamma(u|W)}d\left\{I\{Y \leq u,\Delta_R =1\} -\int_0^u I\{Y \geq v\} d\widehat\Lambda_{T|W}(v|W)\right\} \nonumber\right] \\
    &\quad + \EE_P\left[\frac{(1-S)(\Delta_C- \widehat F_{T|W}(C|W))}{\widehat F_{T|W}(C|W)(1-\widehat F_{T|W}(C|W))}\int_0^C \widehat f_{T|W}(u|W)\widehat h(u,W)du\right] + \EE_P[\widehat\mu(W) - \mu(W)] \\
    &= \underbrace{\EE_P\left[S\int_0^\infty \frac{ \widehat h(u,W) - \widehat\EE[\widehat h(T,W)|T\geq u,W]}{\widehat\Gamma(u|W)}d\left\{I\{Y \leq u,\Delta_R =1\} -\int_0^u I\{Y \geq v\} d\widehat\Lambda_{T|W}(v|W)\right\} \nonumber\right] - \EE_P[S\widehat h(T,W)]}_{\textnormal{term A}}\\
    &\quad + \underbrace{\EE_P[S\widehat h(T,W)] + \EE_P\left[\frac{(1-S)(\Delta_C- \widehat F_{T|W}(C|W))}{\widehat F_{T|W}(C|W)(1-\widehat F_{T|W}(C|W))}\int_0^C \widehat f_{T|W}(u|W)\widehat h(u,W)du\right] + \EE_P[\widehat\mu(W) - \mu(W)]}_{\textnormal{term B}}. \\
\end{align*}

\normalsize 
We will study term A and term B separately. First, we focus on term B. Let $\widehat H(c,w) = \int_0^c \widehat h(t,w)\widehat f(t|w)dt$. We write the expectations as integrals and get
\begin{align*}
    \textnormal{term B} 
    &= \int \pi \widehat h(t,w)f(t|w)p_W(w)dtdw + (1-\pi)\int \frac{F_{T|W}(c|w)\int_0^c \widehat f(t|w)\widehat h(t,w)dt }{\widehat F_{T|W}(c|w)(1-\widehat F_{T|W}(c|w))}g(c|w)p_W(w)dcdw \\
    &\quad - (1-\pi)\int \frac{\widehat F_{T|W}(c|w)\int_0^c \widehat f(t|w)\widehat h(t,w)dt}{\widehat F_{T|W}(c|w)(1-\widehat F_{T|W}(c|w))} g(c|w)p_W(w)dcdw\\
    &\quad + \int \{\widehat \mu(w) - \mu(w)\} p_W(w)dw .
\end{align*}
First, we note that the second term on the right-hand side in the first line of the above display is equivalent to
\begin{align*}
    &\quad (1-\pi)\int \int_0^c \frac{ f(t|w)\widehat H(c,w) }{\widehat F_{T|W}(c|w)(1-\widehat F_{T|W}(c|w))}g(c|w)p_W(w) \ dtdcdw \\
    &= (1-\pi)\int \int_t^\infty \frac{ \widehat H(c,w)g(c|w) }{\widehat F_{T|W}(c|w)(1-\widehat F_{T|W}(c|w))}dc f(t|w) p_W(w) \ dtdw \\
    &= (1-\pi)\int \int_t^\infty \frac{ \widehat H(c,w)\widehat g(c|w) }{\widehat F_{T|W}(c|w)(1-\widehat F_{T|W}(c|w))}dc f(t|w) p_W(w) \ dtdw \\
    &\quad + (1-\pi)\int \int_t^\infty \frac{ \widehat H(c,w)\left\{g(c|w) - \widehat g(c|w)\right\} }{\widehat F_{T|W}(c|w)(1-\widehat F_{T|W}(c|w))}dc f(t|w) p_W(w) \ dtdw.
\end{align*}
Next, we note that $\widehat h$ satisies the following
\begin{equation*}
    \pi \widehat h(t,w)  + (1-\pi)  \int_t^\infty \frac{\widehat H(c,w) \widehat g(c|w)}{\widehat F_{T|W}(c|w)(1-\widehat F_{T|W}(c|w))} dc - \widehat\gamma(w) - I\{t>t^*\} + \widehat\mu(w) = 0.
\end{equation*}
with
\begin{equation*}
    \widehat\gamma(w) = (1-\pi)  \int \int_t^\infty \frac{\widehat H(c,w)\widehat g(c|w)}{\widehat F_{T|W}(c|w)(1-\widehat F_{T|W}(c|w))} \widehat f(t|w)dcdt.
\end{equation*}
Putting these together, we get that
\begin{align*}
    \textnormal{term B} &= \int \left\{\pi \widehat h(t,w) + (1-\pi) \int_t^\infty \frac{ \widehat H(c,w)\widehat g(c|w) }{\widehat F_{T|W}(c|w)(1-\widehat F_{T|W}(c|w))}dc \right\}f(t|w) p_W(w) \ dtdw \\
    &\quad + (1-\pi)\int \int_t^\infty \frac{ \widehat H(c,w)\left\{g(c|w) - \widehat g(c|w)\right\} }{\widehat F_{T|W}(c|w)(1-\widehat F_{T|W}(c|w))}dc f(t|w) p_W(w) \ dtdw \\
    &\quad - (1-\pi)\int \frac{\widehat F_{T|W}(c|w)\int_0^c \widehat f(t|w)\widehat h(t,w)dt}{\widehat F_{T|W}(c|w)(1-\widehat F_{T|W}(c|w))} g(c|w)p_W(w)dcdw\\
    &\quad + \int \{\widehat \mu(w) - \mu(w)\} p_W(w)dw \\
    &= \int \left\{\widehat\gamma(w) + I\{t>t^*\} - \widehat\mu(w) \right\}f(t|w) p_W(w) \ dtdw \\
    &\quad + (1-\pi)\int \int_t^\infty \frac{ \widehat H(c,w)\left\{g(c|w) - \widehat g(c|w)\right\} }{\widehat F_{T|W}(c|w)(1-\widehat F_{T|W}(c|w))}dc f(t|w) p_W(w) \ dtdw \\
    &\quad - (1-\pi)\int \frac{\widehat F_{T|W}(c|w)\int_0^c \widehat f(t|w)\widehat h(t,w)dt}{\widehat F_{T|W}(c|w)(1-\widehat F_{T|W}(c|w))} g(c|w)p_W(w)dcdw\\
    &\quad + \int \{\widehat \mu(w) - \mu(w)\} p_W(w)dw \\
    &= (1-\pi)  \int \int_t^\infty \frac{\widehat H(c,w)\widehat g(c|w)}{\widehat F_{T|W}(c|w)(1-\widehat F_{T|W}(c|w))} \widehat f(t|w)dcdt p_W(w) dw \\
    &\quad + (1-\pi)\int \int_t^\infty \frac{ \widehat H(c,w)\left\{g(c|w) - \widehat g(c|w)\right\} }{\widehat F_{T|W}(c|w)(1-\widehat F_{T|W}(c|w))}dc f(t|w) p_W(w) \ dtdw \\
    &\quad - (1-\pi)\int \frac{\widehat F_{T|W}(c|w)\widehat H(c,w)}{\widehat F_{T|W}(c|w)(1-\widehat F_{T|W}(c|w))} g(c|w)p_W(w)dcdw \\
    &= (1-\pi)  \int \frac{\widehat H(c,w)\widehat g(c|w)}{\widehat F_{T|W}(c|w)(1-\widehat F_{T|W}(c|w))} \widehat F_{T|W}(c|w) p_W(w) dc dw \\
    &\quad + (1-\pi)\int \frac{ \widehat H(c,w)\left\{g(c|w) - \widehat g(c|w)\right\} }{\widehat F_{T|W}(c|w)(1-\widehat F_{T|W}(c|w))}F_{T|W}(t|w) dc  p_W(w) dw \\
    &\quad - (1-\pi)\int \frac{\widehat F_{T|W}(c|w)\widehat H(c,w)}{\widehat F_{T|W}(c|w)(1-\widehat F_{T|W}(c|w))} g(c|w)p_W(w)dcdw \\
    &= (1-\pi)  \int \frac{\widehat H(c,w)\left\{\widehat g(c|w) - g(c|w)\right\}}{1-\widehat F_{T|W}(c|w)} \frac{\left\{\widehat F_{T|W}(c|w) - F_{T|W}(c|w)\right\}}{\widehat F_{T|W}(c|w)}p_W(w) dc dw.
\end{align*}

Next, we turn to term A whose expression is repeated here.
\begin{multline*}
    \EE_P\left[S\int_0^\infty \frac{ \widehat h(u,W) - \widehat\EE[\widehat h(T,W)|T\geq u,W]}{\widehat\Gamma(u|W)}d\left\{I\{Y \leq u,\Delta_R =1\} -\int_0^u I\{Y \geq v\} d\widehat\Lambda_{T|W}(v|W)\right\} \nonumber\right] \\
    - \EE_P[S\widehat h(T,W)].
\end{multline*}
Term A can be re-written using projections onto the space generated by the scores arising from the perturbations to $P_{R|W}$. We use $\Pi_{P}\{\cdot \mid \mathcal{A}\}$ to denote the $L_0^2(P)$-projection operator onto a subspace $\mathcal{A}$ of $L_0^2(P)$, and $\mathcal{T}_{CAR}$ to denote the nuisance tangent space generated by scores of sub-models perturbing the censoring mechanism. Then,
\begin{align*}
   \textnormal{term A} 
    & = \underbrace{\EE_P\left[S\left(\frac{\Delta_R}{\widehat{\Gamma}(T\mid W)} -\frac{\Delta_R}{\Gamma(T\mid W)} \right) \widehat{h}(T,W) \right]}_{\textnormal{(term A-I)}} - \underbrace{\EE_P\left[S\Pi_{\widehat{P}}\left\{ \frac{\Delta_R}{\widehat{\Gamma}(T\mid W)}\widehat{h}(T,W)\mid \mathcal{T_{CAR}}\right\}\right]}_{\textnormal{(term A-II)}}, 
\end{align*}
where we use the fact that under conditional independent censoring, we have $\EE_P[S\widehat h(T,W)] = \EE_P[S\Delta_R\widehat h(T,W)/\Gamma(T|W)]$. We now analyze the two terms separately.
\begin{align*}
    \textnormal{(term A-I)} &:=  \EE_P\left[S\left(\frac{\Delta_R}{\widehat{\Gamma}(T\mid W)} -\frac{\Delta_R}{\Gamma(T\mid W)} \right) \widehat{h}(T,W) \right]\\
    & = \EE_P\left[S \cdot \EE_P\left[\frac{\Delta_R}{\widehat{\Gamma}(T\mid W)} -\frac{\Delta_R}{\Gamma(T\mid W)} \mid T,W,S\right] \widehat{h}(T,W) \right] \\
    & = \EE_P\left[S \cdot \left(\frac{\widehat{\Gamma}(T\mid W)}{\widehat{\Gamma}(T\mid W)} -1 \right) \widehat{h}(T,W) \right] \\
    &  = - \EE_P\left[S \cdot \int_0^T \frac{\Gamma(u\mid W)}{\widehat{\Gamma}(u\mid W)} \left(d\Lambda_R(u\mid W) - d\widehat{\Lambda}_R(u\mid W)\right) \widehat{h}(T,W) \right]\\
    & = - \EE_P\left[S \cdot \int \int_0^t \frac{\Gamma(u\mid W)}{\widehat{\Gamma}(u\mid W)} \left(d\Lambda_R(u\mid W) - d\widehat{\Lambda}_R(u\mid W)\right) \widehat{h}(t,W) f_{T\mid W}(t\mid W) dt\right]\\
    & = - \EE_P\left[S \cdot \int \int \mathbbm{1}(t \geq u) \widehat{h}(t,W) f_{T\mid W}(t\mid W) dt \frac{\Gamma(u\mid W)}{\widehat{\Gamma}(u\mid W)} \left(d\Lambda_R(u\mid W) - d\widehat{\Lambda}_R(u\mid W)\right) \right]\\
    & = - \EE_P\left[S \cdot \int S(u\mid W) \EE_{T\mid W}\left[\widehat{h}(T,W) \mid T \geq u, W\right] \frac{\Gamma(u\mid W)}{\widehat{\Gamma}(u\mid W)} \left(d\Lambda_R(u\mid W) - d\widehat{\Lambda}_R(u\mid W)\right) \right],
\end{align*}
where we denote the conditional cumulative hazard of the censoring mechanism at time $u$ as ${\Lambda}_R(u\mid w)$ and its corresponding estimate as $\widehat{\Lambda}_R(u\mid w)$.

Next we study term A-II. By Theorem 1.1 of \cite{van2003unified}, we have that for any $v(X) \in L_0^2(P)$, the projection operator onto $\mathcal{T}_{CAR}$ takes the form of 
\begin{align*}
    \Pi_P (v\mid \mathcal{T}_{CAR})(x) &= \int_0^Y \left\{ \EE_P[v(X)\mid dA(u) = 1,w,s] - \EE_P[v(X)\mid dA(u) = 0, w,s] \right\}dM_G(u\mid w,s),
\end{align*}
\sloppy where $A(u) = \mathbbm{1}(R \leq u)$ is the indicator of censoring up until time $u$, and $dM_G(u\mid w,s) = \mathbbm{1}(R\in du, \Delta_R = 0) - \mathbbm{1}(Y\geq u) d\Lambda_R(u\mid w,s)$. It is straightforward to show that $\EE_{\widehat{P}}\left[\frac{\Delta_R}{\widehat{\Gamma}(T\mid W)}\widehat{h}(T,W)\mid dA(u) = 1,w,s\right] = 0$. Therefore, we have
\begin{align*}
    \Pi_{\widehat{P}}\left\{ \frac{\Delta_R}{\widehat{\Gamma}(T\mid W)}\widehat{h}(T,W)\mid \mathcal{T_{CAR}}\right\}(x) 
    & = - \int_0^Y  \EE_{\widehat{P}}\left[\frac{\Delta_R}{\widehat{G}(T\mid W)}\widehat{h}(T,W)\mid dA(u) = 0, w,s\right] d\widehat{M}_G(u\mid w,s).
\end{align*}
Note that, 
\begin{align*}
     \EE_{\widehat{P}}\left[\frac{\Delta_R}{\widehat{\Gamma}(T\mid W)}\widehat{h}(T,W)\mid dA(u) = 0, W,S\right]
     & =  \EE_{\widehat{P}}\left[\frac{\Delta}{\widehat{\Gamma}(T\mid W)}\widehat{h}(T,W)\mid T \geq u, R \geq u, W,S\right]\\
     & = \EE_{\widehat{P}}\left[\frac{\mathbbm{1}(T \leq R)}{\widehat{P}(R \geq T\mid W)}\widehat{h}(T,W)\mid T \geq u, R \geq u, W,S\right]\\
     & = \int_{u}^\infty \int_c^\infty \frac{\widehat{h}(t,W)}{\widehat{\Gamma}(t\mid W)} \frac{\widehat{f}_{T\mid W}(t\mid W) \widehat{g}(c\mid W)}{\widehat{S}(t\mid W)\widehat{\Gamma}(t\mid W)}dcdt\\
     & = \int_u^\infty \widehat{h}(t,W)\frac{\widehat{f}_{T\mid W}(t\mid W) }{\widehat{S}(t\mid W)\widehat{\Gamma}(t\mid W)}dt \\
     &  = \frac{\widehat \EE_{T\mid W}\left[\widehat{h}(T,W)\mid T \geq u, W\right]}{\widehat{\Gamma}(u\mid W)}.
\end{align*}
As a result, 
\begin{align*}
    (\text{A-II}) & :=- \EE_P\left[S\Pi_{\widehat{P}}\left\{ \frac{\Delta}{\widehat{\Gamma}(T\mid W)}\widehat{h}(T,W)\mid \mathcal{T_{CAR}}\right\}\right]\\
    & = - \EE_P\left[S\EE_P\left[  \Pi_{\widehat{P}}\left\{ \frac{\Delta}{\widehat{\Gamma}(T\mid W)}\widehat{h}(T,W)\mid \mathcal{T_{CAR}}\right\}(X)  \mid W,S\right] \right] \\
    & = \EE_P\left[S\int \frac{\Gamma(u\mid W)}{\widehat{\Gamma}(u\mid W)}S(u\mid W)\widehat \EE_{T\mid W}\left[\widehat{h}(T,W)\mid T \geq u, W\right] \left(d\Lambda_R(u\mid W) - d\widehat{\Lambda}_R(u\mid W) \right)\right], 
\end{align*}
where the last line is true since 
\begin{align*}
    & \EE_P\left[  \Pi_{\widehat{P}}\left\{ \frac{\Delta}{\widehat{\Gamma}(T\mid W)}\widehat{h}(T,W)\mid \mathcal{T_{CAR}}\right\}(X)  \mid W,S\right]\\
    & =  - \EE_P\left[ \int_0^Y \frac{\widehat \EE_{T\mid W}\left[\widehat{h}(T,W)\mid T \geq u, W\right]}{\widehat{\Gamma}(u\mid W)} d\widehat{M}_G(u\mid W,S) \mid W,S\right]\\
    & = - \int \frac{\Gamma(u\mid W)}{\widehat{\Gamma}(u\mid W)}S(u\mid W)\widehat \EE_{T\mid W}\left[\widehat{h}(T,W)\mid T \geq u, W\right] \left(d\Lambda_R(u\mid W) - d\widehat{\Lambda}_R(u\mid W) \right),
\end{align*}
since
\begin{align*}
    \EE_P\left[d\widehat{M}_G(u\mid W,S) \mid W, S\right] &  = S(u\mid W) \Gamma(u\mid W) \left(d\Lambda_R(u\mid W) - d\widehat{\Lambda}_R(u\mid W) \right).
\end{align*}
Summing up term (A-I) and (A-II), we have
\begin{align*}
    \text{(A-I)} + \text{(A-II)} & = - \EE_P\left[S \cdot \int S(u\mid W) \widehat \EE_{T\mid W}\left[\widehat{h}(T,W)\mid T \geq u, W\right] \frac{\Gamma(u\mid W)}{\widehat{\Gamma}(u\mid W)} \left(d\Lambda_R(u\mid W) - d\widehat{\Lambda}_R(u\mid W)\right) \right] \\
    & \quad  + \EE_P\left[S  \int \frac{\Gamma(u\mid W)}{\widehat{\Gamma}(u\mid W)}S(u\mid W)\widehat \EE_{T\mid W}\left[\widehat{h}(T,W)\mid T \geq u, W\right] \left(d\Lambda_R(u\mid W) - d\widehat{\Lambda}_R(u\mid W) \right)\right]\\
    & = \EE_P\Bigg[S  \int \frac{\Gamma(u\mid W)}{\widehat{\Gamma}(u\mid W)}S(u\mid W) \left(\widehat \EE_{T\mid W}\left[\widehat{h}(T,W)\mid T \geq u, W\right]- \EE_{T\mid W}\left[\widehat{h}(T,W)\mid T \geq u, W\right]\right) \\
    & \hspace{3em} \cdot \left(d\Lambda_R(u\mid W) - d\widehat{\Lambda}_R(u\mid W) \right)\Bigg].
\end{align*}
Equipped with the expansion of the remainder $\Rem(\widehat P, P)$, we are now ready to prove the theorems in Section~\ref{sec: one-step estimation}.

\begin{proof}[Proof of Theorem~\ref{thm:double robust phi}]
    We prove Theorem~\ref{thm:double robust phi} by showing 
    \begin{align}
        \phi_{t^*} & = E_{P}\bigg[\widehat \mu(W) +  \frac{(1-S)(\Delta_{C}-\widehat F_{T|W}(C|W))}{\widehat F_{T|W}(C|W)(1-\widehat F_{T|W}(C|W))}\int_0^{C} \widehat h^*(u,W)d\widehat F_{T|W}(u|W) \nonumber \\
    &\quad + S \int_0^\infty \frac{\widehat h^*(u,W) - \widehat\EE_{T|W}[\widehat h^*(T,W)|T\geq u,W]}{\widehat\Gamma(u|W)}d\left\{I\{Y \leq u,\Delta_{R} =1\} -\int_0^u I\{Y \geq v\} d\widehat\Lambda_{T|W}(v|W)\right\} \bigg],\label{eq:consistency}
    \end{align}
    if either $\widehat F_{T\mid W} = F_{T\mid W}$ (and therefore $\widehat\EE_{T|W}[\widehat h^*(T,W)|T\geq u,W] = \EE_{T|W}[\widehat h^*(T,W)|T\geq u,W]$), or \{$\widehat g = g $ and $\widehat \lambda = \lambda$\}. To show \eqref{eq:consistency}, it suffices to show
    \begin{align*}
        0 &  = E_{P}\bigg[\widehat \mu(W) - \mu(W) +  \frac{(1-S)(\Delta_{C}-\widehat F_{T|W}(C|W))}{\widehat F_{T|W}(C|W)(1-\widehat F_{T|W}(C|W))}\int_0^{C} \widehat h^*(u,W)d\widehat F_{T|W}(u|W) \nonumber \\
        &\quad + S \int_0^\infty \frac{\widehat h^*(u,W) - \widehat\EE_{T|W}[\widehat h^*(T,W)|T\geq u,W]}{\widehat\Gamma(u|W)}d\left\{I\{Y \leq u,\Delta_{R} =1\} -\int_0^u I\{Y \geq v\} d\widehat\Lambda_{T|W}(v|W)\right\} \bigg]\\
        \intertext{Note the right hand side of the above is exactly the remainder term $\Rem(\widehat P, P)$,}
        \mathrm{RHS} & = E_{P}\left[\tau_{\widehat P}(X)  + \widehat \phi_{t^*}  - \phi_{t^*}\right]\\
        & = \Rem(\widehat P, P)\\
        &  = \EE_P\Bigg[S  \int \frac{\Gamma(u\mid W)}{\widehat{\Gamma}(u\mid W)}S(u\mid W) \left(\widehat \EE_{T\mid W}\left[\widehat{h}(T,W)\mid T \geq u, W\right]- \EE_{T\mid W}\left[\widehat{h}(T,W)\mid T \geq u, W\right]\right) \\
        & \hspace{3em} \cdot \left(\lambda(u\mid W) - \widehat{\lambda}(u\mid W) \right)du\Bigg]  \\
        & \quad + (1-\pi)  \int \frac{\widehat H(c,w)\left\{\widehat g(c|w) - g(c|w)\right\}}{1-\widehat F_{T|W}(c|w)} \frac{\left\{\widehat F_{T|W}(c|w) - F_{T|W}(c|w)\right\}}{\widehat F_{T|W}(c|w)}p_W(w) dc dw.\\
        & = 0,
        \end{align*}
        if either $\widehat F_{T\mid W} = F_{T\mid W}$ (and therefore $\widehat\EE_{T|W}[\widehat h^*(T,W)|T\geq u,W] = \EE_{T|W}[\widehat h^*(T,W)|T\geq u,W]$), or \{$\widehat g = g $ and $\widehat \lambda = \lambda$\}. Hence, the double robustness in Theorem~\ref{thm:double robust phi} holds.
\end{proof}

\begin{proof}[Proof of Theorem~\ref{thm:efficiency}]
    As stated previously, we have 
    \begin{align*}
        \widehat \phi_{t^*} - \phi_{t^*} & = (P_n - P)\tau_P + (P_n - P)\{\tau_{\widehat P} - \tau_P\} + \Rem(\widehat P, P). 
    \end{align*}
    Of which, $(P_n - P)\{\tau_{\widehat P} - \tau_P\} = o_p(n^{-1/2})$ if $\widehat{g}$, $\widehat{F}_{T\mid W}$,  and $\widehat{\lambda}_{R\mid W}$ all belong to a fixed Donsker class $\mathcal{F}$ of functions with probability tending to one. Alternatively, if $\widehat{g}$, $\widehat{F}_{T\mid W}$,  and $\widehat{\lambda}$ are obtained via cross-fitting \citep{zheng2010asymptotic}, we also have $(P_n - P)\{\tau_{\widehat P} - \tau_P\} = o_p(n^{-1/2})$. Now let us examine the remainder term
    \begin{align*}
        & \Rem(\widehat P, P)\\
        &  = \EE_P\Bigg[S  \int \frac{\Gamma(u\mid W)}{\widehat{\Gamma}(u\mid W)}S(u\mid W) \left(\widehat \EE_{T\mid W}\left[\widehat{h}(T,W)\mid T \geq u, W\right]- \EE_{T\mid W}\left[\widehat{h}(T,W)\mid T \geq u, W\right]\right) \\
        & \hspace{3em} \cdot \left(\lambda_{R\mid W}(u\mid W) - \widehat{\lambda}_{R\mid W}(u\mid W) \right)du\Bigg]  \\
        & \quad + (1-\pi)  \int \frac{\widehat H(c,w)\left\{\widehat g(c|w) - g(c|w)\right\}}{1-\widehat F_{T|W}(c|w)} \frac{\left\{\widehat F_{T|W}(c|w) - F_{T|W}(c|w)\right\}}{\widehat F_{T|W}(c|w)}p_W(w) dc dw.\\
    \end{align*}
    Note under $\widehat{P}$, $\widehat{g}(u\mid w) = 0$ for any $u \notin [c_l,c_u]$. Then for  any $t > \max(t^*,c_u)$, it is straightforward to show that the unique solution to \eqref{eq: integral equation} has a closed form, that is, $h^*(t,w) = (\mathbbm{1}(t>t^*) - \mu(w))/\pi = (1 - \mu(w))/\pi$. Subsequently, for any $u > \max(t^*,c_u)$, we have
    \begin{align*}
        &\widehat{\EE}_{T\mid W}\left[\widehat{h}(T,W)\mid T \geq u, W\right] - \EE_{T\mid W}\left[\widehat{h}(T,W)\mid T \geq u, W\right] \\
        & = \widehat{\EE}_{T\mid W}\left[\frac{\mathbbm{1}(T>t^*) - \widehat{\mu}(W)}{\pi}\mid T \geq \max(t^*,c_u), W\right] - \EE_{T\mid W}\left[\frac{\mathbbm{1}(T>t^*) - \widehat{\mu}(W)}{\pi}\mid T \geq \max(t^*,c_u), W\right]\\
        & = \widehat{\EE}_{T\mid W}\left[\frac{1 - \widehat{\mu}(W)}{\pi}\mid T \geq \max(t^*,c_u), W\right] - \EE_{T\mid W}\left[\frac{1 - \widehat{\mu}(W)}{\pi}\mid T \geq \max(t^*,c_u), W\right]\\
        & = 0
    \end{align*}
    Hence term A reduces to,
    \begin{align*}
        & \EE_P\Bigg[S  \int \frac{\Gamma(u\mid W)}{\widehat{\Gamma}(u\mid W)}S(u\mid W) \left(\widehat \EE_{T\mid W}\left[\widehat{h}(T,W)\mid T \geq u, W\right]- \EE_{T\mid W}\left[\widehat{h}(T,W)\mid T \geq u, W\right]\right) \\
        & \hspace{3em} \cdot \left(\lambda_{R\mid W}(u\mid W) - \widehat{\lambda}_{R\mid W}(u\mid W) \right)du\Bigg]  \\
        &= \EE_P\Bigg[S  \int_0^{\max(c_i,t^*)} \frac{\Gamma(u\mid W)}{\widehat{\Gamma}(u\mid W)}S(u\mid W) \left(\widehat \EE_{T\mid W}\left[\widehat{h}(T,W)\mid T \geq u, W\right]- \EE_{T\mid W}\left[\widehat{h}(T,W)\mid T \geq u, W\right]\right) \\
        & \hspace{3em} \cdot \left(\lambda_{R\mid W}(u\mid W) - \widehat{\lambda}_{R\mid W}(u\mid W) \right)du\Bigg] 
        &\intertext{Since $\widehat\Gamma(\max(c_u,t^*)\mid W)>\epsilon$, by Cauchy Schwarz inequality, the above is bounded up to a multiplicative factor by }
        &  \int_0^{\max(c_i,t^*)} \lVert \left(\widehat \EE_{T\mid W}\left[\widehat{h}(T,W)\mid T \geq u, W=\cdot\right]- \EE_{T\mid W}\left[\widehat{h}(T,W)\mid T \geq u, W=\cdot\right]\right)\rVert_{L_2(P_W)}\\
        & \hspace{6em} \cdot \|\left(\lambda_{R\mid W}(u\mid \cdot) - \widehat{\lambda}_{R\mid W}(u\mid \cdot) \right)\|_{L_2(P_W)} du
        &\intertext{Under the conditions specified in Theorem~\ref{thm:efficiency}, the above term is }
        & = o_p(n^{-1/2}).
    \end{align*}
    Now we examine term $B$. Note again, $\widehat{g}(u\mid w) = 0$ for any $u \notin [c_l,c_u]$. Term $B$ thus reduces to,
    \begin{align*}
        & (1-\pi)  \int \frac{\widehat H(c,w)\left\{\widehat g(c|w) - g(c|w)\right\}}{1-\widehat F_{T|W}(c|w)} \frac{\left\{\widehat F_{T|W}(c|w) - F_{T|W}(c|w)\right\}}{\widehat F_{T|W}(c|w)}p_W(w) dc dw\\
        & = (1-\pi)  \int_0^{c_u} \frac{\widehat H(c,w)\left\{\widehat g(c|w) - g(c|w)\right\}}{1-\widehat F_{T|W}(c|w)} \frac{\left\{\widehat F_{T|W}(c|w) - F_{T|W}(c|w)\right\}}{\widehat F_{T|W}(c|w)}p_W(w) dc dw
        & \intertext{By Cauchy Schwarz inequality, the above is bounded up to a multiplicative factor by}
        &   \int_0^{c_u} \|\widehat g(c|\cdot) - g(c|\cdot) \|_{L_2(P_W)} \|\widehat F_{T|W}(c|\cdot) - F_{T|W}(c|\cdot)\|_{L_2(P_W)}  dc
        &\intertext{Under the conditions specified in Theorem~\ref{thm:efficiency}, the above term is }
        & = o_p(n^{-1/2}).
    \end{align*}
As a result, we have $\Rem(\widehat P,P) = o_p(n^{-1/2})$. And therefore, 
        \begin{align*}
        \widehat \phi_{t^*} - \phi_{t^*} & = (P_n - P)\tau_P + o_p(n^{-1/2}). 
    \end{align*}
\end{proof}

Before proving Theorem~\ref{thm:efficiency real}, we need to analyze a similar remainder term but defined with the canonical gradient $\tau_P\eff$. With a slight abuse of notation, we still denote it as $\textnormal{Rem}(\widehat P, P)$. In the following, all the expectations are taken with respect to the observation unit while treating all (estimated) nuisance functions as fixed. We have the remainder term in the asymptotic linear expansion of our one-step estimator
\begin{align*}
    \textnormal{Rem}(\widehat P, P) &= \EE_P\Bigg[S \left\{\Delta_R \widehat\eta^*(Y,W) - \int_0^Y \widehat\eta^*(t,W)\widehat\lambda_{T|W}(t|W) dt\right\} + \\
    &\quad \quad (1-S)\left\{\frac{\Delta_C - \widehat F_{T|W}(C|W)}{\widehat F_{T|W}(C|W)}\right\}\widehat\Theta^*(C,W) + \widehat\mu(W)-\mu(W)\Bigg].
\end{align*}
We first compute the expectation of the terms related to the right-censored data only. 
\begin{align*}
    &\quad \EE_P\left[S \left\{\Delta_R \widehat\eta^*(Y,W) - \int_0^Y \widehat\eta^*(t,w)\widehat\lambda_{T|W}(t|w) dt\right\}\right] \\
    &= \pi \EE_P[\EE_P[\Delta_R \mid T,W,S=1]\widehat\eta^*(T,W) \mid S=1] \\
    &\quad - \pi \EE_P\left[\int \EE[I\{Y \geq t\}|W,S=1]\widehat\eta^*(t,W)\widehat\lambda_{T|W}(t|W) dt \mid S=1\right] \\
    &= \pi \EE_P[\Gamma(T|W)\widehat\eta^*(T,W) \mid S=1] - \pi \EE_P\left[\int \Gamma(t|W)S_{T|W}(t|W)\widehat\eta^*(t,W)\widehat\lambda_{T|W}(t|W) dt \mid S=1\right] \\
    &= \pi \int \int \Gamma(t|w)S_{T|W}(t|w)\widehat\eta^*(t,w)\left\{\lambda_{T|W}(t|w) - \widehat\lambda_{T|W}(t|W)\right\} p_W(w) dtdw.
\end{align*}
At the same time, the expectation of the term related to the current status data only is
\begin{align*}
    &\quad \EE_P\left[(1-S)\left\{\frac{\Delta_C - \widehat F_{T|W}(C|W)}{\widehat F_{T|W}(C|W)}\right\}\widehat\Theta^*(C,W)\right] \\
    &= (1-\pi)\EE_P\left[\frac{F_{T|W}(C|W) - \widehat F_{T|W}(C|W)}{\widehat F_{T|W}(C|W)}\widehat\Theta^*(C,W) \mid S=0\right] \\
    &= (1-\pi)\int \int \frac{F_{T|W}(c|w) - \widehat F_{T|W}(c|w)}{\widehat F_{T|W}(c|w)}\widehat\Theta^*(c,w) g(c|w)p_W(w)dcdw.
\end{align*}
Thus, the remainder term is
\begin{align*}
    &\quad\textnormal{Rem}(\widehat P, P) \\
    &= \pi \int \int \Gamma(t|w)S_{T|W}(t|w)\widehat\eta^*(t,w)\left\{\lambda_{T|W}(t|w) - \widehat\lambda_{T|W}(t|w)\right\} p_W(w) dtdw \\
    &\quad + (1-\pi)\int \int \frac{F_{T|W}(c|w) - \widehat F_{T|W}(c|w)}{\widehat F_{T|W}(c|w)}\widehat\Theta^*(c,w) g(c|w)p_W(w)dcdw \\
    &\quad +\int \{\widehat\mu(w)-\mu(w)\}p_W(w)dw \\
    &= \pi \int \int \left\{\Gamma(t|w)S_{T|W}(t|w) - \widehat\Gamma(t|w)\widehat S_{T|W}(t|w)\right\}\widehat\eta^*(t,w)\left\{\lambda_{T|W}(t|w) - \widehat\lambda_{T|W}(t|w)\right\} p_W(w) dtdw \\
    &\quad + \pi \int \int \widehat\Gamma(t|w)\widehat S_{T|W}(t|w)\widehat\eta^*(t,w)\left\{\lambda_{T|W}(t|w) - \widehat\lambda_{T|W}(t|w)\right\} p_W(w) dtdw \\
    &\quad + (1-\pi)\int \int \frac{F_{T|W}(c|w) - \widehat F_{T|W}(c|w)}{\widehat F_{T|W}(c|w)}\widehat\Theta^*(c,w) g(c|w)p_W(w)dcdw \\
    &\quad +\int \{\widehat\mu(w)-\mu(w)\}p_W(w)dw \\
    &= \pi \int \int \left\{\Gamma(t|w)S_{T|W}(t|w) - \widehat\Gamma(t|w)\widehat S_{T|W}(t|w)\right\}\widehat\eta^*(t,w)\left\{\lambda_{T|W}(t|w) - \widehat\lambda_{T|W}(t|w)\right\} p_W(w) dtdw \\
    &\quad - \int \int  (1-\pi) \left\{ \int_t^\infty \frac{\widehat S_{T|W}(c|w)\widehat\Theta^*(c,w)}{\widehat F_{T|W}(c|w)} \widehat g(c|w) dc \right\}
    \left\{\lambda_{T|W}(t|w) - \widehat\lambda_{T|W}(t|w)\right\} p_W(w) dtdw \\
    &\quad + (1-\pi)\int \int \frac{F_{T|W}(c|w) - \widehat F_{T|W}(c|w)}{\widehat F_{T|W}(c|w)}\widehat\Theta^*(c,w) g(c|w)p_W(w)dcdw \\
    &\quad +\int \{\widehat\mu(w)-\mu(w)\}p_W(w)dw - \int \int  \widehat\mu(w)I\{t \leq t^*\}\left\{\lambda_{T|W}(t|w) - \widehat\lambda_{T|W}(t|w)\right\} p_W(w) dtdw \\
    &= \pi \int \int \left\{\Gamma(t|w)S_{T|W}(t|w) - \widehat\Gamma(t|w)\widehat S_{T|W}(t|w)\right\}\widehat\eta^*(t,w)\left\{\lambda_{T|W}(t|w) - \widehat\lambda_{T|W}(t|w)\right\} p_W(w) dtdw \\
    &\quad - \int \int  (1-\pi) \frac{\widehat S_{T|W}(c|w)\widehat \Theta^*(c,w)}{\widehat F_{T|W}(c|w)} \widehat g(c|w) \left\{\Lambda_{T|W}(c|w) - \widehat\Lambda_{T|W}(c|W)\right\} p_W(w) dcdw \\
    &\quad + (1-\pi)\int \int \frac{F_{T|W}(c|w) - \widehat F_{T|W}(c|w)}{\widehat F_{T|W}(c|w)}\widehat\Theta^*(c,w) g(c|w)p_W(w)dcdw \\
    &\quad +\int \{\widehat\mu(w)-\mu(w)\}p_W(w)dw - \int \int  \widehat\mu(w)I\{t \leq t^*\}\left\{\lambda_{T|W}(t|w) - \widehat\lambda_{T|W}(t|w)\right\} p_W(w) dtdw
\end{align*}
Further splitting the terms, we get
\begin{align*}
&\quad\textnormal{Rem}(\widehat P, P) \\
    &=\underbrace{\pi \int \int \left\{\Gamma(t|w)S_{T|W}(t|w) - \widehat\Gamma(t|w)\widehat S_{T|W}(t|w)\right\}\widehat\eta^*(t,w)\left\{\lambda_{T|W}(t|w) - \widehat\lambda_{T|W}(t|w)\right\} p_W(w) dtdw}_{\textnormal{term 1}} \\
    &\quad - \underbrace{(1-\pi) \int \int \frac{\widehat S_{T|W}(c|w)\widehat \Theta^*(c,w)}{\widehat F_{T|W}(c|w)} \left\{\widehat g(c|w) - g(c|w)\right\} \left\{\Lambda_{T|W}(c|w) - \widehat\Lambda_{T|W}(c|w)\right\} p_W(w) dcdw}_{\textnormal{term 2}} \\
    &\quad + \underbrace{(1-\pi)\int \int \frac{F_{T|W}(c|w) - \widehat F_{T|W}(c|w) - \widehat S_{T|W}(c|w)\left\{\Lambda_{T|W}(c|w) - \widehat\Lambda_{T|W}(c|w)\right\}}{\widehat F_{T|W}(c|w)}\widehat\Theta^*(c,w) g(c|w)p_W(w)dcdw}_{\textnormal{term 3}} \\
    &\quad +\underbrace{\int \left(\widehat\mu(w)-\mu(w) - \widehat\mu(w)\left\{\Lambda_{T|W}(t^*|w) - \widehat\Lambda_{T|W}(t^*|w)\right\}\right) p_W(w) dw}_{\textnormal{term 4}}.
\end{align*}

Both term 1 and term 2 are already second-order. We now study terms 3 and 4. Noting that $\Lambda = -\log S$, we re-write the integrand in term 4 and apply a Taylor expansion
\begin{align*}
    &\quad \widehat\mu(w)-\mu(w) + \widehat\mu(w)\left\{\log \mu(w) - \log \widehat \mu(w)\right\} \\
    &= \widehat\mu(w)-\mu(w) + \widehat\mu(w)\left\{\frac{1}{\widehat\mu(w)}(\mu(w) - \widehat \mu(w)) - \frac{1}{2\bar\mu(w)^2}(\mu(w) - \widehat \mu(w))^2\right\} \\
    &= - \frac{\widehat\mu(w)}{2\bar\mu(w)^2}(\mu(w) - \widehat \mu(w))^2
\end{align*}
for some value $\bar\mu(w)$ between $\mu(w)$ and $\widehat\mu(w)$. Hence, term 4 is equal to
\begin{equation*}
    - \int \frac{\widehat\mu(w)}{2\bar\mu(w)^2}(\mu(w) - \widehat \mu(w))^2 p_W(w) dw
\end{equation*}
which is second-order provided that both $\mu(w)$ and $\widehat\mu(w)$ are lower bounded away from 0.

We handle the numerator in the fraction in the integrand of term 3 in a similar fashion. Specifically,
\begin{align*}
    &\quad F_{T|W}(c|w) - \widehat F_{T|W}(c|w) + \widehat S_{T|W}(c|w)\left\{\log S_{T|W}(c|w) - \log \widehat S_{T|W}(c|w)\right\} \\
    &= F_{T|W}(c|w) - \widehat F_{T|W}(c|w) + \\
    &\quad + \widehat S_{T|W}(c|w)\left\{\frac{1}{\widehat S_{T|W}(c|w)} (S_{T|W}(c|w) - \widehat S_{T|W}(c|w)) - \frac{1}{2\bar S_{T|W}(c|w)^2}(S_{T|W}(c|w) - \widehat S_{T|W}(c|w))^2\right\} \\
    &= F_{T|W}(c|w) - \widehat F_{T|W}(c|w) + S_{T|W}(c|w) - \widehat S_{T|W}(c|w) - \frac{\widehat S_{T|W}(c|w)}{2\bar S_{T|W}(c|w)^2}(S_{T|W}(c|w) - \widehat S_{T|W}(c|w))^2 \\
    &= - \frac{\widehat S_{T|W}(c|w)}{2\bar S_{T|W}(c|w)^2}(S_{T|W}(c|w) - \widehat S_{T|W}(c|w))^2
\end{align*}
for some $\bar S_{T|W}(c|w)$ between $\widehat S_{T|W}(c|w)$ and $S_{T|W}(c|w)$. Thus, term 3 is proportional to
\begin{equation*}
    -\int \int \frac{\widehat S_{T|W}(c|w)(S_{T|W}(c|w) - \widehat S_{T|W}(c|w))^2}{2\bar S_{T|W}(c|w)^2\widehat F_{T|W}(c|w)}\widehat\Theta^*(c,w) g(c|w)p_W(w)dcdw
\end{equation*}
which is second-order provided that $S(c|w)$ is bounded away from 0 on the support of $C$ and both $\widehat F_{T|W}(c|w)$ and $\widehat S_{T|W}(c|w)$ are bounded away from 0.

\begin{proof}[Proof of Theorem~\ref{thm:efficiency real}]
To begin with, we point out that the nuisance functions involved in constructing $\widehat \phi^{\text{eff}}_{t^*}$ are as follows: $\widehat{\eta^*}$, $\widehat \lambda_{T\mid W}$, $\widehat F_{T\mid W}$, $\widehat \Gamma$, $\widehat S_{T\mid W}$, $\widehat g$ and $\widehat \mu$. Since the estimated nuisance functions must be compatible under $\widehat P$,  $\widehat \mu$, $\widehat F_{T\mid W}$ and $\widehat S_{T\mid W}$ are all dictated by $\widehat \lambda_{T\mid W}$. Subsequently, the list of nuisance functions reduces to $\widehat{\eta^*}$, $\widehat \lambda_{T\mid W}$, $\widehat \lambda_{R\mid W}$ and  $\widehat g$. We will later show that the estimation uncertainty that $\widehat{\eta^*}$ introduces can be fully accounted by the estimation errors of $\widehat \lambda_{T\mid W}$, $\widehat \lambda_{R\mid W}$ and  $\widehat g$. As stated previously, we have
    \begin{align*}
        \widehat \phi^{\text{eff}}_{t^*} - \phi_{t^*} & = (P_n - P)\tau^{\text{eff}}_P + (P_n - P)\{\tau^{\text{eff}}_{\widehat P} - \tau_P\} + \Rem(\widehat P, P). 
    \end{align*}
    We first examine the remainder term. Based on previous derivations, we have
    \begin{align*}
        & \text{Rem}(\widehat P, P) \\
        & = 
\pi \int \int \left\{\Gamma(t|w)S_{T|W}(t|w) - \widehat\Gamma(t|w)\widehat S_{T|W}(t|w)\right\}\widehat\eta^*(t,w)\left\{\lambda_{T|W}(t|w) - \widehat\lambda_{T|W}(t|w)\right\} p_W(w) dtdw \\
    &\quad - (1-\pi) \int \int \frac{\widehat S_{T|W}(c|w)\widehat \Theta^*(c,w)}{\widehat F_{T|W}(c|w)} \left\{\widehat g(c|w) - g(c|w)\right\} \left\{\Lambda_{T|W}(c|w) - \widehat\Lambda_{T|W}(c|w)\right\} p_W(w) dcdw\\
    &\quad + (1-\pi)\int \int \frac{F_{T|W}(c|w) - \widehat F_{T|W}(c|w) - \widehat S_{T|W}(c|w)\left\{\Lambda_{T|W}(c|w) - \widehat\Lambda_{T|W}(c|w)\right\}}{\widehat F_{T|W}(c|w)}\widehat\Theta^*(c,w) g(c|w)p_W(w)dcdw\\
    &\quad +\int \left(\widehat\mu(w)-\mu(w) - \widehat\mu(w)\left\{\Lambda_{T|W}(t^*|w) - \widehat\Lambda_{T|W}(t^*|w)\right\}\right) p_W(w) dw.
    \end{align*}
    It is straightforward to verify for any $t > \max (c_u, t^*)$, the unique solution to \eqref{eq: integral equation with rc} is $\eta^* = 0$. In addition, under $\widehat P$, $\widehat{g}(u\mid w) - g(u\mid w) = 0$ for any $u \notin \mathcal{C}$. Hence the remainder term reduces to,
    \begin{align*}
        & \text{Rem}(\widehat P, P) \\
        & = \pi \int \int_0^{\max(c_u, t^*)} \left\{\Gamma(t|w)S_{T|W}(t|w) - \widehat\Gamma(t|w)\widehat S_{T|W}(t|w)\right\}\widehat\eta^*(t,w)\left\{\lambda_{T|W}(t|w) - \widehat\lambda_{T|W}(t|w)\right\} p_W(w) dtdw \\
    &\quad - (1-\pi) \int \int_{c_l}^{c_u} \frac{\widehat S_{T|W}(c|w)\widehat \Theta^*(c,w)}{\widehat F_{T|W}(c|w)} \left\{\widehat g(c|w) - g(c|w)\right\} \left\{\Lambda_{T|W}(c|w) - \widehat\Lambda_{T|W}(c|w)\right\} p_W(w) dcdw\\
    &\quad + (1-\pi)\int \int_{c_l}^{c_u} \frac{F_{T|W}(c|w) - \widehat F_{T|W}(c|w) - \widehat S_{T|W}(c|w)\left\{\Lambda_{T|W}(c|w) - \widehat\Lambda_{T|W}(c|w)\right\}}{\widehat F_{T|W}(c|w)}\widehat\Theta^*(c,w) g(c|w)p_W(w)dcdw\\
    &\quad +\int \left(\widehat\mu(w)-\mu(w) - \widehat\mu(w)\left\{\Lambda_{T|W}(t^*|w) - \widehat\Lambda_{T|W}(t^*|w)\right\}\right) p_W(w) dw.
    \end{align*}
    By Cauchy Schwarz inequality, the first term is bounded up to a multiplicative factor by
    \begin{align*}
        \int_0^{\max(c_u, t^*)} \left\{\lVert  \Gamma(t\mid \cdot) - \widehat \Gamma(t\mid \cdot) \rVert_2 + \lVert  S_{T\mid W}(t\mid \cdot) - \widehat S_{T\mid W}(t\mid \cdot) \rVert_{L_2(P_W)}  \right\} \lVert  \lambda_{T\mid W}(t\mid \cdot) - \widehat \lambda_{T\mid W}(t\mid \cdot) \rVert_{L_2(P_W)} dt.
    \end{align*}
    The second term is bounded up to a multiplicative factor by
    \begin{align*}
        \int_{c_l}^{c_u} \lVert \widehat g(c|\cdot) - g(c|\cdot)\rVert_{L_2(P_W)} \lVert  \Lambda_{T\mid W}(c\mid \cdot) - \widehat \Lambda_{T\mid W}(c\mid \cdot) \rVert_{L_2(P_W)} dc.
    \end{align*}
    The third term is bounded up to a multiplicative factor by
    \begin{align*}
               \int_{c_l}^{c_u} \lVert  S_{T\mid W}(c\mid \cdot) - \widehat S_{T\mid W}(c\mid \cdot) \rVert^2_{L_2(P_W)} dc.
    \end{align*}
    The fourth term is bounded up to a multiplicative factor by
    \begin{align*}
                \lVert  \mu(\cdot) - \widehat \mu(\cdot)  \rVert^2_{L_2(P_W)} .
    \end{align*}
    Under conditions specified in Theorem~\ref{thm:efficiency real}, that is, each of the nuisance functions $\widehat g$, $\widehat \lambda_{T\mid W}$ and $\widehat \lambda_{R\mid W}$ is being estimated at $o_p(n^{-1/4})$ rate, all of the four terms above are $o_p(n^{-1/2})$. 
    
    Lastly, since  $\widehat{\eta^*}$ does not introduce additional estimation errors, $(P_n - P)\{\tau^{\text{eff}}_{\widehat P} - \tau^{\text{eff}}_P\} = o_p(n^{-1/2})$ if $\widehat{g}$, $\widehat{\lambda}_{T\mid W}$,  and $\widehat{\lambda}_{R\mid W}$ all belong to a fixed Donsker class $\mathcal{F}$ of functions with probability tending to one. Alternatively, if $\widehat{g}$, $\widehat{\lambda}_{T\mid W}$,  and $\widehat{\lambda}_{R \mid W}$ are obtained via cross-fitting \citep{zheng2010asymptotic}, we also have $(P_n - P)\{\tau^{\text{eff}}_{\widehat P} - \tau^{\text{eff}}_P\} = o_p(n^{-1/2})$. 
\end{proof}

\subsection{Proof of Lemma~\ref{lemma: EIF covariate shift}}
\begin{proof}
Let us again start with the intermediate observation unit $X^I = (S,W,ST,(1-S)C,(1-S)\Delta_C)$ with distribution $P^I$ belonging to the model $\calM_2^I = \{\mathbb{P}_2(\Pi,\widetilde P_{0,W},\widetilde P_{1,W}, \widetilde P_{T|W},\widetilde P_{C|W})\}$ for some functional $\mathbb{P}_2$. The tangent space at $P^I$ with respect to the model $\calM_2^I$ is as follows.
\begin{align}
    \calT_2^I &= \overline{\textnormal{span}}\Bigg\{sh(t,w) + (1-s)h_C(c,w) + \frac{(1-s)(\delta_C - F_{T|W}(c|w))}{F_{T|W}(c|w)(1-F_{T|W}(c|w))}\int_0^c f_{T|W}(u|w)h(u,w)du \nonumber \\
    &\quad + s h_{1,W}(w) + (1-s)h_{0,W}(w), h_{i,W} \in L_2^0(P_{i,W}) \textnormal{ for } i \in \{0,1\}, h(t,w) \in L_2^0(P_{T|W}), h_C(c,w) \in L_2^0(P_{C|W})\Bigg\}.
\end{align}
This can be derived in the same way as in the proof of Lemma~\ref{lemma: tangent space FO and CS} except that $P_{0,W}$ and $P_{1,W}$ need to be perturbed separately, and the details are omitted.

Define the functionals $\Psi_1:\calM_2^I \rightarrow \mathbb{R}$ and $\Psi_0:\calM_2^I \rightarrow \mathbb{R}$ such that $\Psi_i(P^I) = \EE_{P_{i,W}}[P_{T|W}(T>t^*|W)]$ for $i \in \{0,1\}$. Note that the functional $\Psi_i$ does not depend on $P_{C|W}$ and $P_{1-i,W}$. Thus, for the purpose of finding the canonical gradient, it suffices to consider a reduced model where $P_{1-i,W}$ and $P_{C|W}$, and to find elements $h_{i,W}^* \in L_2^0(P_{i,W})$ and $h^* \in L_2^0(P_{T|W})$ that satisfies the equation that defines a gradient. The rest of the proof is similar to that of Lemma~\ref{lemma: EIF FO and CS}. 

Fix a value of $i \in \{0,1\}$. First, consider a submodel $\{P_\epsilon^I:\epsilon\}$ such that the density of $P_{\epsilon, i,W}$ is $p_{\epsilon,i,W} = p_{i,W}(w)(1+\epsilon h_{i,W}(w))$ and all other component distributions are the same as under $P^I$. For notational convenience, we write $\Psi_i(P^I)$ as $\mu_i$, and let $\mu(w)$ again denote $P_{T|W}(T>t^*|W=w)$. Then,
\begin{align*}
    \frac{d\Psi_i(P_\epsilon^I)}{d\epsilon} &= \int \left( I\{t > t^* \} - \mu_i \right)  f(t|w) p_{i,W}(w) h_{i,W}(w) dtdw \\
    &= \int (\mu(w)-\mu_i) h_{i,W}(w)p_{i,W}(w)dw \\
    &= \EE_{P^I}\Bigg[\mathbbm{1}(S=i)h_{i,W}(W)\times \\
    &\quad \left\{\mathbbm{1}(S=i)h_{i,W}^*(W) + Sh^*(T,W) + \frac{(1-S)(\Delta_C - F_{T|W}(C|W))}{F_{T|W}(C|W)(1-F_{T|W}(C|W))}\int_0^C f(u|W)h^*(u,W)du\right\}\Bigg] \\
    &= \EE_{P^I}\left[\mathbbm{1}(S=i)h_{i,W}(W)h_{i,W}^*(W) \right] \\
    &= \Pi(S=i)\EE_{P_{i,W}}\left[h_{i,W}(W)h_{i,W}^*(W) \right] \\
    &=  \Pi(S=i) \int h_{i,W}(w)h_{i,W}^*(w)p_{i,W}(w)dw.
\end{align*}
The above holds for all $h_{i,W} \in L_2^0(P_{i,W})$, and therefore, $h_{i,W}^*(w) = (\mu(w) - \mu_i)/\Pi(S=i)$.  

Next, we consider parametric submodels that perturb $P_{T|W}$ and find $h^*(t,w)$. Specifically, we now consider a submodel $\{P_\epsilon^I:\epsilon\}$ such that the conditional density of $T$ given $W$ under $P_\epsilon^I$ is $f(t|w)(1+\epsilon h(t,w))$.
\begin{align*}
    \frac{d\Psi_i(P_\epsilon^I)}{d\epsilon} &= \int \left( I\{t>t^*\} - \mu(w)\right) h(t,w) f(t|w)  p_{i,W}(w) dtdw \\
    &= \EE_{P^I}\Bigg[\left\{\mathbbm{1}(S=i)h_{i,W}^*(W) + Sh^*(T,W) + \frac{(1-S)(\Delta_C - F_{T|W}(C|W))}{F_{T|W}(C|W)(1-F_{T|W}(C|W))}\int_0^C f(u|W)h^*(u,W)du\right\} \\
    &\quad \qquad \times \left\{Sh(T,W) + \frac{(1-S)(\Delta_C - F_{T|W}(C|W))}{F_{T|W}(C|W)(1-F_{T|W}(C|W))}\int_0^C f(u|W)h(u,W)du\right\}\Bigg] \\
    &= \EE_{P^I}[Sh(T,W)h^*(T,W)] \\
    &\quad + \EE_{P^I}\left[(1-S)\frac{(\Delta_C - F_{T|W}(C|W))^2}{F_{T|W}(C|W)^2(1-F_{T|W}(C|W))^2}\int_0^C f(v|W)h(v,W)dv\int_0^C f(u|W)h^*(u,W)du\right] \\
    &= \pi \int h(t,w)h^*(t,w) f(t|w)p_{1,W}(w)dtdw \\
    &\quad + (1-\pi)\int \frac{H^*(c,w)}{F_{T|W}(c|w)(1-F_{T|W}(c|w))} \left(\int_0^c f(t|w)h(t,w)dt \right) g(c|w) p_{0,W}(w) dcdw \\
    &= \pi \int h(t,w)h^*(t,w) f(t|w)\frac{p_{1,W}(w)}{p_{i,W}(w)}p_{i,W}(w)dtdw \\
    &\quad + (1-\pi)\int \left( \int_t^\infty \frac{H^*(c,w) g(c|w)}{F_{T|W}(c|w)(1-F_{T|W}(c|w))} dc \right) h(t,w) f(t|w) \frac{p_{0,W}(w)}{p_{i,W}(w)}p_{i,W}(w)\ dtdw \\
    &= \int \left\{ \pi \frac{p_{1,W}(w)}{p_{i,W}(w)} h^*(t,w)  + (1-\pi) \frac{p_{0,W}(w)}{p_{i,W}(w)} \int_t^\infty \frac{H^*(c,w) g(c|w)}{F_{T|W}(c|w)(1-F_{T|W}(c|w))} dc \right\} h(t,w) f(t|w) p_{i,W}(w)dtdw
\end{align*}
The first term in the curly brackets has been 0 conditional on $W$, and the second term can be centered with $\gamma(w)$. As $h$ is arbitrary, the following must hold for almost every $(t,w)$,
\begin{equation}\label{eq: integral equation covariate shift repeat}
   \pi \frac{p_{1,W}(w)}{p_{i,W}(w)} h^*(t,w)  + \frac{p_{0,W}(w)}{p_{i,W}(w)}\left\{(1-\pi) \int_t^\infty \frac{H^*(c,w) g(c|w)}{F_{T|W}(c|w)(1-F_{T|W}(c|w))}dc -\gamma(w) \right\} - I\{t>t^*\} + \mu(w) = 0.
\end{equation}

Given the derivations above, we obtain the canonical gradient of $\Phi_i$ evaluated at $P^I$ relative to the model $\calM_2^I$ as follows:
\begin{equation*}
    \xi_{i,P^I}: x^I \mapsto sh^*(t,w) + \frac{(1-s)(\delta_C - F_{T|W}(c|w))}{F_{T|W}(c|w)(1-F_{T|W}(c|w))}\int_0^c f_{T|W}(u|w)h^*(u,w)du + \frac{\mathbbm{1}(s=i)}{\Pi(S=i)}\left\{\mu(w) - \mu_i\right\},
\end{equation*}
where $h^*$ solves the equation in \eqref{eq: integral equation covariate shift repeat}. Finally, applying the results for coarsening-at-random, for example, that in \citet{van2007note}, we obtain a valid gradient of $\Phi_i$ evaluated at the observed data distribution $P$ relative to the observed data model $\calM_2$ is
\begin{align*}
    \xi_{i,P}: x &\mapsto s\int_0^\infty \frac{h^*(u,w) - \EE_{T|W}[h^*(T,W)|T\geq u,W=w]}{\Gamma(u|w)}d\left\{I\{y \leq u,\delta_R =1\} -\int_0^u I\{y \geq v\} d\Lambda_{T|W}(v|w)\right\} \nonumber \\
    &\quad + \frac{(1-s)(\delta_C-F_{T|W}(c|w))}{F_{T|W}(c|w)(1-F_{T|W}(c|w))}\int_0^c f_{T|W}(u|w)h^*(u,w)du + \frac{\mathbbm{1}(s=i)}{\Pi(S=i)}\left\{\mu(w) - \mu_i\right\}.
\end{align*}
The canonical gradient can be obtained by projecting a gradient based on right-censored data alone onto the observed data tangent space.
\end{proof}

\end{document}